\def\useLongLipics{}
\newif\iflong
\newif\ifshort
\long\def\inlong#1{\iflong#1\fi}
\long\def\inshort#1{\ifshort#1\fi}

\ifdefined\useLongLipics
   \longtrue
   \shortfalse

   \documentclass[a4paper,USenglish]{lipics-v2019}

\else
   \longfalse
   \shorttrue

   \documentclass[a4paper,USenglish,anonymous]{lipics-v2019}

\fi

\usepackage[WOenumitem,WOprobsubsuperscripts,WOhyperref]{auxlib}
\usepackage{xspace}
\usepackage{float}
\usepackage{multirow}
\usepackage{nicefrac}
\usepackage[abbreviations,binary-units]{siunitx}

\usepackage[ruled,vlined]{algorithm2e} %
\Crefname{algocf}{Algorithm}{Algorithms}

\usepackage[
        doi         = true,
        giveninits  = true,
        isbn        = false,
        maxbibnames = 99,
        maxcitenames = 2,
        sortcites   = true,
        style       = numeric,
        url         = false,
    ]{biblatex}

\defbibheading{bibliography}{%
	\section*{References}
}

\AtEveryBibitem{
    \clearfield{month}
    \clearlist{location}
    \clearlist{publisher}
    \clearname{editor}
}

\iflong
\makeatletter
\let\lst@floatdefault\lst@float
\makeatother
\fi

\usepackage{intcalc}
\usepackage{tikz}
\usetikzlibrary{decorations.pathmorphing,calc,arrows.meta,arrows}
\tikzset{snake it/.style={decorate, decoration=snake}}
\tikzset{onif/.code args={<#1>#2}{\ifthenelse{#1}{\pgfkeysalso{#2}}{}}}
\tikzset{
	conflict/.style={DarkRed, fill=red!10},
	state1/.style={DarkGreen, fill=green!10},
	state2/.style={DarkBlue, fill=blue!5},
	state3/.style={Maroon, fill=Yellow!10},
	state4/.style={Red, fill=Thistle!30},
	state5/.style={Black, fill=black!5},
	state?/.style={conflict}
}
\def\arrowHead{Latex[length=0.5em, width=0.5em]}

\nolinenumbers

\title{Simulating Population Protocols in Sub-Constant Time per Interaction}
\makeatletter
\def\@author{}
\makeatother
\iflong
\author{Petra Berenbrink}{Universität Hamburg, Germany}{petra.berenbrink@uni-hamburg.de}{}{}
\author{David Hammer}{University of Southern Denmark \and Goethe University Frankfurt, Germany}{hammer@imada.sdu.dk}{https://orcid.org/0000-0002-0226-3475}{}
\author{Dominik Kaaser}{Universität Hamburg, Germany}{dominik.kaaser@uni-hamburg.de}{https://orcid.org/0000-0002-2083-7145}{}
\author{Ulrich Meyer}{Goethe University Frankfurt, Germany}{umeyer@ae.cs.uni-frankfurt.de}{}{}
\author{Manuel Penschuck}{Goethe University Frankfurt, Germany}{mpenschuck@ae.cs.uni-frankfurt.de}{}{}
\author{Hung Tran}{Goethe University Frankfurt, Germany}{hung@ae.cs.uni-frankfurt.de}{}{}

\authorrunning{P.\ Berenbrink, D.\ Hammer, D.\ Kaaser, U.\ Meyer, M.\ Penschuck, and H.\ Tran}
\Copyright{Petra Berenbrink, David Hammer, Dominik Kaaser, Ulrich Meyer, Manuel Penschuck, and Hung Tran}
\funding{%
   This work was partially supported by the Deutsche Forschungsgemeinschaft (DFG) under grants ME~2088/3-2, ME~2088/4-2, and ME 2088/5-1.%
}
\acknowledgements{%
   This material is based upon work initiated on a workshop of the DFG~FOR~2971/1 ``Algorithms, Dynamics, and Information Flow in Networks''.
   We would like to thank the Center for Scientific Computing, University of Frankfurt for making their HPC facilities available.%
}
	\supplement{%
		Implementations of the simulators (including scripts and data for reproduciblity) are available at \url{https://ae.cs.uni-frankfurt.de/r/projects/population-simulator}.
}
\else
	\ccsdesc{}
	\def\subjclassHeading{\vskip-0.7\baselineskip}
	\supplement{%
		Customizable implementations of the simulators discussed (including scripts and data to reproduce figures included in this paper) are referenced in the full version of this paper.
}
\fi


\hideLIPIcs %

\makeatletter
\def\?#1{}
\def\whp{\hc{w.h.p}\@ifnextchar.{.\?}{\@ifnextchar,{.}{\@ifnextchar){.}{.\ }}}}
\def\Whp{\hc{W.h.p}\@ifnextchar.{.\?}{\@ifnextchar,{.}{.\ }}}
\makeatother

\let\epsilon\varepsilon

\newcommand{\hc}[1]{%
   {#1}%
}
\newcommand{\hcm}[1]{\hc{\ensuremath{#1}}\xspace}
\newcommand{\hcs}[1]{\hc{#1}\xspace}
\newcommand{\formatIdentifier}[1]{\hcm{\mathsf{#1}}}

\newcommand{\colDistr}[2]{\ensuremath{\formatIdentifier{coll}\left(#1, #2\right)}}

\def\dat{Dynamic Alias Table\xspace}

\def\untouched {\formatIdentifier{untouched}}
\def\delayed   {\formatIdentifier{delayed}}
\def\updated   {\formatIdentifier{updated}}

\def\Oh        {\hcm{\operatorname{O}}}
\def\oh        {\hcm{\operatorname{o}}}
\def\distrAs   {\sim}
\def\ellHalf   {\hcm{\floor{\ell/2}}}
\def\q         {\hcm{|Q|}}

\def\eg        {\hcs{e.g.,}}
\def\cf        {\hcs{cf.}\ }
\def\ie        {\hcs{i.e.,}}

\def\outDomain{\hcm{Y}}
\def\outFn{\hcm{\gamma}}

\newcommand{\confAtTime}[1]{\hcm{C(#1)}}
\def\initialConf{\hcm{C_0}}

\newcommand{\algStepWiseBase}[1]{\hcm{\text{\textsc{Seq}}_\mathsf{#1}}}

\def\algStepWise{\algStepWiseBase{}}

\def\algStepWiseArray{\algStepWiseBase{Array}}
\def\algStepWiseArrayPre{\hcm{\text{\textsc{Seq}}_\mathsf{Array}^\mathsf{prefetch}}}
\def\algStepWiseBST  {\algStepWiseBase{BST}}
\def\algStepWiseAlias{\algStepWiseBase{Alias}}
\def\algStepWiseLinear{\algStepWiseBase{Linear}}

\def\algSimple{\hcm{\text{\textsc{Batched}}}}
\def\algMulti{\hcm{\text{\textsc{MultiBatched}}}}

\renewcommand{\set}[1]{\hcm{\left\{#1\right\}}}
\newcommand{\problem}[1]{\textsc{#1}}

\ifshort
\makeatletter
\renewcommand\subsection{\@startsection{subsection}{2}{\z@}%
                                     {-2ex\@plus -1ex \@minus -.2ex}%
                                     {1ex \@plus .2ex}%
                                     {\sffamily\Large\bfseries\raggedright}}
\def\paragraph#1{\@startsection{subparagraph}{5}{\z@}%
                                       {1ex \@plus .25ex \@minus .5ex}%
                                       {-1em}%
                                      {\sffamily\normalsize\bfseries}{#1.}}
\let\dcm@maketitle\maketitle
\def\maketitle{\dcm@maketitle\markright{}\setcounter{page}{0}}
\makeatother
\SetAlFnt{\footnotesize}
\fi

\makeatletter%
\begin{document}

   \maketitle

   \begin{abstract}
We consider the problem of efficiently simulating population protocols.
In the population model, we are given a distributed system of $n$ agents modeled as identical finite-state machines.
In each time step, a pair of agents is selected uniformly at random to interact.
In an interaction, agents update their states according to a common transition function.
We empirically and analytically analyze two classes of simulators for this model.
First, we consider sequential simulators executing one interaction after the other.
Key to the performance of these simulators is the data structure storing the agents' states.
For our analysis, we consider plain arrays, binary search trees, and a novel \dat data structure.
Secondly, we consider batch processing to efficiently update the states of multiple independent agents in one step.
For many protocols considered in literature, our simulator requires amortized sub-constant time per interaction and is fast in practice:
given a fixed time budget, the implementation of our batched simulator is able to simulate population protocols several orders of magnitude larger compared to the sequential competitors,
and can carry out $2^{50}$ interactions among the same number of agents in less than \SI{400}{\second}.    \end{abstract}

   \inshort{
      \clearpage
      \setcounter{page}{1}
   }

\section{Introduction} \label{sec:introduction}
We consider the \emph{population model}, introduced by \textcite{DBLP:journals/dc/AngluinADFP06} to model systems of resource-limited mobile agents that interact to solve a common task.
Agents are modeled as finite-state machines.
The computation of a \emph{population protocol} is a sequence of pairwise interactions of agents.
In each interaction, the two participating agents observe each other's states and update their own state according to a transition function common to all agents.
\iflong
    Note that population protocols do not halt.
\fi

Typical applications of population protocols are networks of passively mobile sensors \cite{DBLP:conf/podc/AngluinADFP04}.
As an example, consider a flock of birds, where each bird is equipped with a simple sensor.
Two sensors communicate whenever their birds are sufficiently close.
An application could be a distributed disease monitoring system raising an alarm if the number of birds with high temperature rises above some threshold.
Further processes which resemble properties of population protocols include chemical reaction networks \cite{DBLP:journals/nc/SoloveichikCWB08}, programmable chemical controllers at the level of DNA \cite{chen2013}, or biochemical regulatory processes in living cells \cite{cardelli2012}.

\iflong
    Early theoretical results \cite{DBLP:conf/podc/AngluinAE06,DBLP:journals/dc/AngluinAER07} focused on the computational power of the population model.
    For example, it has been shown \cite{DBLP:conf/podc/AngluinADFP04,DBLP:journals/dc/AngluinADFP06} that all semilinear predicates can be computed by a population protocol.
\fi
While the computational power of population protocols with constantly many states per agent is well understood by now (see below), less is known about the power of protocols with state spaces growing with the population size.
In this setting, much interest has been on analyzing the runtime and state space requirements for \emph{probabilistic} population protocols, where the two interacting agents are sampled in each time step independently and uniformly at random from the population.
This notion of a probabilistic scheduler allows the definition of a \emph{runtime} of a population protocol.
The runtime and the number of states are the main performance measures used in the theoretical analysis of population protocols.

For the theoretical analysis of population protocols, a large toolkit is available in the literature.
\iflong
    These tools range from standard techniques, such as tail bounds, potential functions, and couplings of Markov chains to mean field approximations and differential equations (see \cref{sec:related-work} for an overview of related work).
\fi
Consequently, the remaining gaps between upper and lower bounds for many quantities of interest have been narrowed down:
for many protocols, the required number of states has become sub-logarithmic, while the runtime approaches more and more the (trivial) lower bounds for any meaningful protocol.
So far, when designing new protocols, simulations have always proven a versatile tool in getting an intuition for these stochastic processes.
However, once observables are of order $\log\log{n}$ and below, naive population protocol simulators fail to deliver the necessary insights (\eg $\log\log{n} \le 5$ for typical input sizes of $n \le 2^{32}$).
Our main contribution in this paper is a new simulation approach allowing to execute a large number of interactions even if the population size exceeds $2^{40}$.
In the remainder of this section, we first give a formal model definition in \cref{sec:model} and then describe our main contributions and related work in \cref{sec:our-contributions,sec:related-work}.

\subsection{Formal Model Definition}
\label{sec:model}
In the \emph{population model}, we are given a distributed system of~$n$ agents modeled as finite-state machines.
A \emph{population protocol} is specified by a state space $Q = \set{q_1, \ldots, q_{\q}}$, an output domain $\outDomain$, a transition function $\delta\colon Q\times Q \rightarrow Q\times Q$, and an output function $\outFn\colon Q \rightarrow O$.
At time~$t$, each agent $i$ has a state $s_i(t) \in Q$, which is updated during the execution of the protocol.
The current output of agent $i$ in state $s_i(t)$ is $\outFn(s_i(t))$.
The \emph{configuration} $\confAtTime{t} = \set{s_1(t), \ldots, s_n(t)}$ of the system at time $t$ contains the states of the agents after $t$ interactions.
For the sake of readability, we omit the parameter~$t$ in $\confAtTime{t}$ and $s_i(t)$ when it is clear from the context.
The \emph{initial configuration} is denoted $\confAtTime{0} = \initialConf$.

The \emph{computation} of a population protocol runs in a sequence of discrete time steps.
In each time step, a \emph{probabilistic scheduler} selects an ordered pair of agents $(u, v)$ independently and uniformly at random to \emph{interact}.
Agent $u$ is called the \emph{initiator} and agent $v$ is the \emph{responder}.
During this \emph{interaction}, both agents $u$ and $v$ observe each other's state and update their states according to the transition function $\delta$ such that $\left(s_u(t{+}1),\ s_v(t{+}1)\right) \gets \delta\left(s_u(t),\ s_v(t)\right)$.

A given problem for the population model specifies the agents' initial states, the output domain $O$, and (a set of) \emph{desired (output) configurations} for a given input.
As an example, consider the \problem{Majority} problem.
Each agent is initially in one of two states $q_A$ and $q_B$ corresponding to two opinions $A$ and $B$.
Assuming that $A$ is the initially dominant opinion, the protocol concludes once all agents $u$ give $\outFn(s_u) = A$ as their output.
Any configuration in which all agents output the initially dominant opinion is a desired configuration.

This notion of a desired configuration allows to formally define two notions of a \emph{runtime} of a population protocol.
The \emph{convergence time} $T_C$ is the number of interactions until the system enters a desired configuration and never leaves the desired configurations in a given run.
The \emph{stabilization time $T_S$} is the number of interactions until the system enters a desired \emph{stable} configuration for which there does not exist any sequence of interactions due to which the system leaves the desired configurations.
A population protocol is \emph{stable}, if it always eventually reaches a desired output configuration.%

A number of variants of this model are commonly used.
For \emph{symmetric protocols}, the order of the interacting agents is irrelevant for the transition.
In particular, this means that if $\delta(q_u, q_v) = (q_u', q_v')$, then $\delta(q_v, q_u) = (q_v', q_u')$.
In \emph{protocols with probabilistic transition functions}, the outcome of an interaction may be a random variable.
In \emph{one-way protocols}, %
only the initiator updates its states such that $\delta(q_u, q_v) = (q_u', q_v)$ for any interaction.

\paragraph{Model Assumptions}
We assume a \emph{meaningful} protocol which converges after at most $N = \poly(n)$ interactions, has an $\Oh(1)$ time transition function~$\delta$, and uses $\q < \sqrt{n}$ states (observe that many relevant protocols only use $\q = \Oh(\polylog n)$ states; see \cref{sec:related-work}).

\subsection{Our Contributions}
\label{sec:our-contributions}
In this paper, we present a new approach for simulating population protocols.
Our simulator allows us to efficiently simulate a large number, $N$, of interactions for large populations of size $n$.
Our findings are summarized in \cref{tab:simulators}. \inshort{The proofs are given in the appendix.}
\begin{table}[t]
	\centering
    \caption{%
       Simulating $N$ interactions among $n$~agents in $\q$~states.
       For \algMulti, we restrict  $\q = \omega(\sqrt{\log n})$. %
       Values indicated by $\dagger$ hold in expectation.
    }
    \label{tab:simulators}
	\def\thformat#1{\sffamily\bfseries #1}
	\def\th#1{\multicolumn{1}{c|}{\thformat{#1}}}
        \begin{tabular}{|l|l|l|l|l|}\cline{2-5}
           \th{}&\th{Simulator}&\th{Section}&\th{Time Complexity}&\th{Space Complexity (bits)}\\\cline{2-5}\noalign{\vskip\doublerulesep}\cline{1-5}
			\multirow{4}{*}{\rotatebox[origin=c]{90}{\thformat{Sequential}}}
            & \algStepWiseArray  & \cref{sec:sequential-simulation} & $\Theta(N)$             & $\Theta(n \log \q)$ \\\cline{2-5}
            & \algStepWiseLinear & \cref{sec:sequential-simulation} & $\Oh(N \q)$ & $\Theta(\q \log n)$ \\\cline{2-5}
            & \algStepWiseBST    & \cref{sec:sequential-simulation} & $\Theta(N \log \q)$     & $\Theta(\q \log n)$ \\\cline{2-5}
            & \algStepWiseAlias  & \cref{sec:sequential-simulation} & $\Theta(N)$ \whp        & $\Theta(\q \log n)$ \\\cline{1-5}\noalign{\vskip\doublerulesep}\cline{1-5}
			\multirow{2}{*}{\rotatebox[origin=c]{90}{\thformat{Batch}}}
            & \algSimple        & \cref{sec:batch-processing} & $\Oh(N(\log n + \q^2) / \sqrt{n})^\dagger$ & $\Theta(\q \log n)$ \\\cline{2-5}
            & \algMulti         & \cref{sec:merging-runs}     & $\Oh(N \q \sqrt{\log n} / \sqrt{n} )^\dagger$ & $\Theta( \q \log n)$ \\\hline
        \end{tabular}
\end{table}

\paragraph{Sequential Simulators}
As a baseline, we directly translate the population model into a sequential algorithm framework~$\algStepWise$:
\algStepWise selects for each interaction two agents uniformly at random, updates their states, and repeats.
We analyze the runtime and memory consumption of various variants in~\cref{sec:sequential-simulation}.

\paragraph{Batch Processing}
To speed up the simulation, we introduce and exploit \emph{collision-free runs}, a sequence of interactions where no agent participates more than once.
Our algorithms \algSimple and \algMulti coalesce these independent interactions into batches for improved efficiency.
\algSimple first samples the length $\ell$ of a collision-free run.
It then randomly pairs $\ell$ independent agents, adds one more interaction ---the collision--- reusing one of the run's agents, and finally repeats.
\algSimple is presented in \cref{sec:batch-processing} and extended into \algMulti in \cref{sec:merging-runs}.
We discuss practical details and heuristics in \cref{sec:heuristics}.

\paragraph{Dynamic Alias Tables}
The simulation of population protocols often needs an \emph{urn-like} data structure to efficiently sample random agents (marbles) and update their states (colors).
The \emph{alias method} \cite{DBLP:journals/toms/Walker77,DBLP:journals/tse/Vose91} enables random sampling from arbitrary discrete distributions in $\Oh(1)$ time.
However, it is static in that the distribution may not change over time.
Thus, we extend it and analyze a \emph{\dat} in \cref{sec:sequential-simulation}.
It supports sampling with and without replacement uniformly at random (u.a.r.) and addition of elements (if the urn is sufficiently full).
We believe this data structure might be of independent interest and show:

\begin{theorem}[name=,restate=restateThmDynamicAliasTable,label=thm:dynamic-alias-table]
Let $U$ be a \dat that stores an urn of $n$ marbles, where each marble has one of $k$ possible colors.
$U$ requires $\Theta(k \log{n})$ bits of storage. If $n \geq k^2$, we can \begin{itemize}
\item select a marble u.a.r.\ from $U$ with replacement in expected constant time,
\item select a marble u.a.r.\ from $U$ without replacement in expected amortized constant time,
\item and add a marble of a given color to $U$ in amortized constant time.
\end{itemize}
\end{theorem}

\subsection{Related Work}
\label{sec:related-work}

The population model was introduced in \cite{DBLP:conf/podc/AngluinADFP04,DBLP:journals/dc/AngluinADFP06}, assuming a constant number of states per agent.
Together with \cite{DBLP:conf/podc/AngluinAE06,DBLP:journals/dc/AngluinAER07}, their results show that all semilinear predicates are stably computable in this model.
In the following, we focus on two prominent problems, \problem{Majority} and \problem{Leader Election}.
For a broad overview, we refer to surveys \cite{DBLP:journals/eatcs/AspnesR07} and \cite{DBLP:journals/eatcs/ElsasserR18}.

In \cite{DBLP:journals/dc/AngluinAE08} a \problem{Majority} protocol with three states is presented where the agents agree on the majority after $\LDAUOmicron{n \log n}$ interactions \whp (\emph{with high probability} $1-n^{-\Omega(1)}$), if the initial numbers of agents holding each opinion differ by at least $\LDAUomega{\sqrt{n} \log n}$.
In \cite{DBLP:conf/icalp/MertziosNRS14,DBLP:journals/siamco/DraiefV12}, four-state protocols are analyzed that stabilize in expectation in $\LDAUOmicron{n^2 \log n}$ interactions.
In a recent series of papers~\cite{DBLP:conf/nca/MocquardAABS15, DBLP:conf/soda/AlistarhAEGR17, DBLP:conf/soda/AlistarhAG18, DBLP:conf/podc/AlistarhGV15, DBLP:conf/podc/BilkeCER17, DBLP:conf/wdag/BerenbrinkEFKKR18, DBLP:journals/corr/abs-1805-04586}, bounds for the \problem{Majority} problem have been gradually improved.
The currently best known protocol \cite{DBLP:conf/wdag/BerenbrinkEFKKR18} solves \problem{Majority} \whp in $\LdauOmicron{n \log^2 n}$ interactions using $\LdauOmicron{n \log^{5/3} n}$ states.
Regarding lower bounds, \cite{DBLP:conf/soda/AlistarhAEGR17} shows that
protocols with less than $(\log\log n)/2$ states require in expectation $\LDAUOmega{n^2 / \polylog(n)}$ interactions to stabilize.
In \cite{DBLP:conf/soda/AlistarhAG18} it is shown that any \problem{Majority} protocol that stabilizes in $n^{2 - \LDAUOmega{1}}$ expected interactions requires $\LDAUOmega{\log n}$ states under some natural \inlong{monotonicity }assumptions.

The goal for \problem{Leader Election} protocols is that exactly one agent is in a designated leader state.
\Textcite{DBLP:conf/wdag/DotyS15} show that any population protocol with a constant number of states that stably elects a leader requires $\LDAUOmega{n^2}$ expected interactions, a bound matched by a natural two-state protocol.
Upper bounds for protocols with a non-constant number of states per agent were presented in~\cite{DBLP:conf/icalp/AlistarhG15, DBLP:conf/soda/AlistarhAEGR17, DBLP:conf/podc/BilkeCER17, DBLP:conf/soda/AlistarhAG18, DBLP:conf/soda/BerenbrinkKKO18, DBLP:conf/soda/GasieniecS18, DBLP:conf/spaa/GasieniecSU19, berenbrink2020}.
In \cite{DBLP:conf/soda/GasieniecS18} a \problem{Leader Election} protocol that stabilizes \whp in $\LdauOmicron{n  \log^2 n}$ interactions, using $\LdauOmicron{\log\log n}$ states (matching a corresponding lower bound~\cite{DBLP:conf/soda/AlistarhAEGR17}) is presented.
The core idea is to synchronize the agents using a \emph{phase-clock}.
The currently best known protocol for \problem{Leader Election} is due to \cite{berenbrink2020}, stabilizing in expected $\LdauOmicron{n \log n}$ interactions using $\LdauOmicron{\log\log{n}}$ states per agent.

As a tool for self-synchronization, so-called \emph{phase-clocks} have been explored in a wide range of related areas, see, \eg the seminal paper \cite{DBLP:journals/ppl/AroraDG91}.
In the population model, the concept of phase-clocks was first introduced in \cite{DBLP:journals/dc/AngluinAE08a} under the assumption that a leader is present.
These clocks were generalized in \cite{DBLP:conf/soda/GasieniecS18} to a \emph{junta} of $n^{\epsilon}$ agents.
In \cref{sec:simulations} we empirically analyze a variant of this phase-clock process.
\section{Sequential Simulation}
\label{sec:sequential-simulation}
\ifshort
\begin{algorithm}[b]
\else
\begin{algorithm}[t]
\fi
   \caption{\algStepWise: The algorithmic framework for sequential simulation.}
   \label{alg:seq-algo}

   \DontPrintSemicolon
   \SetKwInOut{Input}{input\expandafter\?\?}

   \Input{~%
      configuration $C$,
      transition function $\delta$,
      number of steps $N$
   }

   \For{$t \gets 1$ \KwTo $N$}{
    sample and remove agents $i$ and $j$ without replacement from $C$\;
    add agents in states $\delta(s_i, s_j)$ to $C$\;
   }
\end{algorithm}
 
\iflong
In the following we consider the simulation of population protocols with $n$ agents.
Given some initial configuration \initialConf, our goal is to simulate a protocol over a large number~$N$ of steps with $N \gg n$ in order to eventually obtain the final configuration~$\confAtTime N$.
\fi

As a baseline, we first consider variants of \algStepWise, a sequential approach defined in \cref{alg:seq-algo}.
It is a direct translation of the machine model discussed in \cref{sec:model}.
\algStepWise carries out $N$ steps in a fully serialized manner.
For each interaction, it selects two agents uniformly at random, computes their new states based on their current ones, and updates the configuration.

Under the realistic assumption that the transition function~$\delta$ can be evaluated in constant time, \algStepWise's runtime and memory footprint is dominated by storing, sampling from, and updating the configuration~$C$.
We therefore consider appropriate data structures.
In the population model, agents typically are anonymous, \ie we cannot distinguish two agents in the same state.
Hence, we can store a configuration~$C$ as an unordered multiset $\hat C$ and maintain multiplicities rather than individual states.\iflong\footnote{
   There also exists model variants where interactions are limited to some underlying communication network, resulting in a restricted interaction graph.
   In this setting, agents may become distinguishable based on the network structure, and we may no longer describe a configuration as a multiset of states.
   See \cref{sec:open-problems} for further discussions.
}\fi\xspace
To this end, \algStepWise requires an \emph{urn-like} data structure which efficiently supports (i) weighted sampling (with and without replacement) and (ii) adding of single agents.
In the following, we consider various data structures and their impact on the complexity of the sequential approach.

\paragraph{Array}
\algStepWiseArray maintains the configuration $C$ in an array $A[1 \ldots n]$ where $A[i]$ holds $s_i$, the state of the $i$-th agent.
Sampling with replacement is trivial, as we only draw a uniform variate $X \in [n]$ and return $A[X]$.
Sampling without replacement works analogously: we overwrite $A[i]$ with $A[n]$ and remove the array's last element $A[n]$.
Adding new elements is possible by appending. (Note that we do not grow the memory since we always store at most $n$ agents in the array.)
This leads to an $\Oh(N)$ time algorithm and a memory footprint of $\Oh(n \log \q)$ bits, which can be prohibitively large if simulating large populations in parallel.

\paragraph{Linear Search}
\algStepWiseLinear maintains the multiset $\hat C$ in an array~$A$ such that $A[i]$ holds the number of agents in state~$q_i$.
Sampling requires a linear search on~$A$ in $\Oh(\q)$ per sample.
This results in a worst-case simulation time of $\Theta(N\q)$.
Nevertheless, in practice \algStepWiseLinear is among the fastest sequential variants for small $\q$ (see \cref{sec:simulations}).
Compared to \algStepWiseArray, it has a significantly smaller memory footprint of $\Oh(\q \log{n})$ bits.

\paragraph{Binary Search Tree}
\algStepWiseBST maintains the multiset $\hat C$ using a balanced binary search tree.
The $i$-th leaf (from left to right) encodes $\hat C_i$, the number of agents in state~$i$.
Each inner node~$v$ stores the number~$\ell_v$ of agents in its left subtree.
To randomly sample an agent, we draw an integer $X$ from $\set{0, \dots, n-1}$ uniformly at random and compare it to the root's value~$\ell_r$.
If $X < \ell_r$, the sample is in the interval covered by the left sub-tree, and we descend accordingly.
Otherwise, we update $X \gets X - \ell_r$ and descend into the right subtree.
We recurse until some leaf $i$ is reached, where we emit an agent of state~$i$.

Each operation on the tree involves a simple path from the root to a leaf of length $\Theta(\log \q)$.
Since the work per level is constant, all operations take $\Theta(\log \q)$ time.
Thus, \algStepWiseBST requires $\Theta(N \log \q)$ total time and $\Oh(\q \log n)$ bits of memory.

\paragraph{\dat{s}}
\begin{figure}
    \centering
    \scalebox{0.8}{%
\begin{tikzpicture}[
   first/.style={fill=green!20, opacity=0.9},
   second/.style={fill=yellow!40, opacity=0.9},
   reject/.style={fill=red!40, opacity=0.9}
]
   \def\dx{*1em}
   \def\dy{*1.5em}

   \def\tabX{20}

   \node[anchor=south] (Fi) at (\tabX \dx + 1, 0) {$F[i]$};
   \node[anchor=south] (Si) at (\tabX \dx + 3\dx, 0) {$S[i]$};
   \node[anchor=south] (Ai) at (\tabX  \dx+ 5\dx, 0) {$A[i]$};
   \node[anchor=south, font=\footnotesize, align=center] (Rr) at (\tabX \dx + 8.5 \dx, 0) {rejection \\ probability};

   \node[first,  minimum width=2\dx, minimum height=5\dy, anchor=north west] at (\tabX \dx, 0) {};
   \node[second, minimum width=4\dx, minimum height=5\dy, anchor=north west] at (\tabX \dx + 2\dx, 0) {};
   \node[reject, minimum width=5\dx, minimum height=5\dy, anchor=north west] at (\tabX \dx + 6\dx, 0) {};

   \foreach \x in {0, ..., 10} {\path[draw, black] (\x \dx, 0) to (\x \dx, -5 \dy);}
   \foreach \x in {0, 2, 4, 6, 11} {\path[draw, black] (\tabX \dx + \x \dx, 0) to (\tabX \dx + \x \dx, -5 \dy);}

   \foreach \qi/\ni/\qii/\nii [count=\i] in {1/7/4/2, 2/5/1/1, 3/0/1/3, 4/6/1/0, 5/4/1/2}
   {
      \ifnum\ni>0
         \node[first, minimum width=\ni \dx, minimum height=1\dy, inner sep=0, anchor=north west] at (0, -\intcalcDec{\i} \dy) {$q_\qi$};
      \fi
      \ifnum\nii>0
         \node[second, minimum width=\nii \dx, minimum height=1\dy, inner sep=0, anchor=north west] at (\ni \dx, -\intcalcDec{\i} \dy) {$q_\qii$};
          \path[draw] (\ni \dx, -\intcalcDec{\i} \dy) to ++(0, -1\dy);
     \fi
      \def\ns{\intcalcAdd{\ni}{\nii}}
      \node[reject, minimum width=9 \dx - \ns \dx, minimum height=1\dy, inner sep=0, anchor=north west] at (\ns \dx, -\intcalcDec{\i} \dy) {};

      \node[anchor=center] at (\tabX\dx + 1 \dx, -\intcalcDec{\i} \dy - 0.5 \dy) {\ni};
      \node[anchor=center] at (\tabX\dx + 3 \dx, -\intcalcDec{\i} \dy - 0.5 \dy) {\nii};
      \node[anchor=center] at (\tabX\dx + 5 \dx, -\intcalcDec{\i} \dy - 0.5 \dy) {$q_\qii$};
      \node[anchor=center] at (\tabX\dx + 8.5 \dx, -\intcalcDec{\i} \dy - 0.5 \dy) {%
         \ifthenelse{\equal{\ns}{9}}{$0$}{$\nicefrac{\intcalcSub{9}{\ns}}{9}$}%
      };
   }

   \foreach \y in {0, ..., 5} {
      \path[draw, black] (     0, -\y \dy) to (10 \dx, -\y \dy);
      \path[draw, black] (\tabX \dx, -\y \dy) to (\tabX \dx + 11 \dx, -\y \dy);
   }

   \path[draw, black, thick] (6 \dx, 0) to ++(0, -5.5 \dy) node[below] {$\frac n k$};
   \path[draw, blue, very thick] (3 \dx, 0) to ++(0, -5.5 \dy) node[below] {$\alpha \frac n k$};
   \path[draw, blue, very thick] (10 \dx, 0) to ++(0, -5.5 \dy) node[below] {$\beta \frac n k$};
   \path[draw, red, thick] (3 \dx, 0.5\dy) node[above] {$R_\text{min}$} to ++(0, -5.5 \dy);
   \path[draw, red, thick] (9 \dx, 0.5\dy) node[above] {$R_\text{max}$}to ++(0, -5.5 \dy);

   \node[inner sep=0.1em] (reject) at (13\dx, -3\dx) {reject};
   \path[draw, bend left, -o] (reject) to (8\dx, -3.5\dx);

\end{tikzpicture} %
}
    \caption{%
      \dat storing $\hat C = \left(q_1\colon 13,\ q_2\colon 5,\ q_3\colon 0,\ q_4\colon 8,\ q_5\colon 4 \right)$,
      \ie $n = 30$ and $k = 5$.
      This imbalanced configuration will soon need rebuilding, \eg after the next decrease of $q_1$ in row~3, or after adding two more agents in state~$q_1$ in row~1.%
    }
    \label{fig:alias-table}
\end{figure}
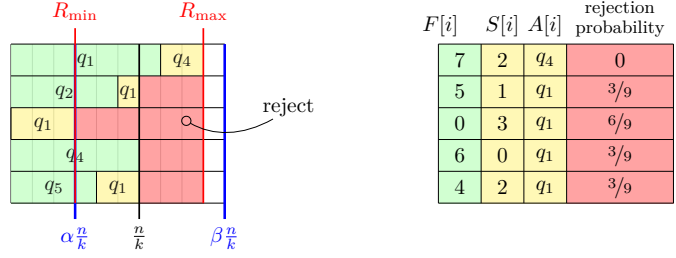

\algStepWiseAlias combines the linear runtime of \algStepWiseArray (\whp) with the small memory footprint of \algStepWiseBST, provided $\q < \sqrt n$.
At the heart of \algStepWiseAlias lies a \emph{\dat}.
This data structure encodes an urn that contains $n$ marbles, each colored with one of $k$ possible colors.
It allows us to sample marbles with and without replacement and to add new marbles of a given color in expected amortized constant time.
In the following, we present the details of the \dat and sketch a proof of \cref{thm:dynamic-alias-table}.
The full proof can be found in \cref{apx:dynamic-alias-table}.

The \dat is based on the so-called \emph{alias method} introduced by \textcite{DBLP:journals/toms/Walker77}.
The alias method allows sampling random variates $X$ with $\prob{X = i} = p_i$ for $1 \leq i \leq k$ for arbitrary finite discrete probability distributions $p_1,\dots,p_k$.
It requires two arrays, $F[1 \dots k]$ and $A[1 \dots k]$, which can be precomputed in $\Oh(k)$ time \cite{DBLP:journals/tse/Vose91}.

The two arrays define a table with $k$ rows and two entries per row.
Each row represents an equal probability mass of $1 / k$.
The first entry of the $i$-th row corresponds to element~$i$.
It is assigned a \emph{weight} $0 \le F[i] \le 1$.
The second entry, the so-called \emph{alias}, is given by $A[i]$ and has the remaining weight $1 - F[i]$.
To sample from the table, we first select row~$i$ uniformly at random.
Then, we draw a uniform variate~$X$ from $[0, 1)$.
We return element~$i$ if $X < F[i]$ (\ie with probability $F[i]$).
Otherwise, we return the alias $A[i]$ (\ie with probability $1 - F[i]$).

The \dat directly processes and stores the multiplicities of $\hat C$ as integers  rather than operating on real-valued probabilities.
In general and due to rounding errors, $n$ agents cannot be evenly distributed over $k = \q$ rows.
Hence, we introduce a second array~$S[1\dots k]$ storing the aliases' weights (see \cref{fig:alias-table}).
It is straightforward to generalize the table construction in \cite{DBLP:journals/tse/Vose91}  to our setting while keeping the original runtime of $\Oh(k)$.

Our data structure now allows us to sample elements without replacement.
Whenever we sample and remove an element, we decrement its weight using the counter of the row we sampled from.
This is always possible, since only elements with strictly positive weights can be sampled in the first place.
It also allows us to insert a new agent with state $q_i$ by simply incrementing the weight~$F[i]$ of the first element in row~$i$.

Let $R[i] = F[i] + S[i]$ denote the weight of row~$i$, and define $R_\text{min}$ and $R_\text{max}$ as smallest and largest row weights, respectively.
In contrast to the original alias method, our row weights may differ (\ie $R_\text{min} \ne R_\text{max}$).
Hence, the original sampling procedure overrespresents rows with weights smaller than $R_\text{max}$.
We remove this bias using \emph{rejection sampling} as follows.
We first select row~$i$ uniformly at random.
Then we draw a uniform variate~$X$ from $\set{0, \dots, R_\text{max}{-}1}$.
If $X < F[i]$, we emit the first element $i$, and if $F[i] \le X < R[i]$, we return the second element $A[i]$.
Otherwise, the trial is rejected and we restart the sampling process.
This ensures that the probability of returning an element from row~$i$ is $R[i] / n$.

The expected runtime complexity of sampling is $\Oh(f)$, where $f$ is the expected rejection rate with $f \le f' = R_\text{max} / R_\text{min}$.
In order to maintain a sampling time which is constant in expectation, we limit the ratio~$f'$ as follows.
After each update, we ensure that $\alpha \floor{n/k} \le R_\text{min} \le R_\text{max} \le \beta \ceil{n/k}$.
Otherwise, we rebuild the data structure.
Here, $\alpha < 1$ and $\beta > 1$ are parameters chosen such that $\beta / \alpha = \Oh(1)$.

\section{Batch Processing}
\label{sec:batch-processing}

So far, we discussed algorithms to simulate a population protocol step-by-step.
These simulators can output the population's configuration $\confAtTime{t}$ for each time step~$1 \le t \le N$.
With a time complexity of $\Oh(N)$, the simulators \algStepWiseArray and \algStepWiseAlias are optimal in this sense.
In practice, however, it often suffices to obtain a configuration snapshot every $\Theta(n)$ steps.
In this setting, we can achieve sub-constant work per interaction under mild assumptions. %

\iflong
\begin{figure}
    \centering
    \scalebox{0.95}{%
\begin{tikzpicture}
   \def\is{*4.2em}

   \node[anchor=west] at (0em, 5em) {\textbf{The original interaction sequence (\cf \cref{{sec:sequential-simulation}}):}};
   \foreach \a/\b/\k/\l [count=\i] in {3/2/1/2, 1/2/3/4, 1/1/5/6, 2/1/7/8, 3/3/9/10, 1/2/11/12, d/d/d/d, 1/1/{2\ell-1}/{2\ell}, 2/?/{2\ell+1}/6} {
      \ifthenelse{\equal{\a}{d}}{
         \node at (\i \is, 0) {$\dots$};
      }{
         \node[draw, circle, inner sep=0, minimum height=1.5em, state\a] (rqa\i) at (\i \is - 0.22 \is, 0) {$q_\a$};
         \node[draw, circle, inner sep=0, minimum height=1.5em, state\b] (rqb\i) at (\i \is + 0.22 \is, 0) {$q_\b$};
         \node[anchor=south] at (rqa\i.north) {$s_{i_{\k}}$};
         \node[anchor=south, onif={<\equal{\l}{6}> conflict}] at (rqb\i.north) {$s_{i_{\l}}$};
      }
   }

   \node[draw, minimum width=0.5\is, minimum height=0.5\is, anchor=north] (trans) at (3\is, -0.5\is) {$\delta$};
   \path[draw, semithick, green, arrows={-\arrowHead}] (rqa3) to (trans);
   \path[draw, semithick, green, arrows={-\arrowHead}] (rqb3) to (trans);

   \path[draw, red, semithick, arrows={-\arrowHead}] (trans.east) -| node[pos=0.25, black, font=\small, align=center] {
      the updated state of the agent drawn twice is\\
      known only after the $\delta$ was evaluated} (rqb9);
   \path[draw, red, semithick, arrows={-\arrowHead}, bend right] (trans.east) to (rqb3);
   \path[draw, red, semithick, arrows={-\arrowHead}, bend left] (trans.west) to (rqa3);

   \path[draw] ($(rqa1.west) + (-0.3em, 0)$) to ++(0, 2.5em) -| node [pos=0.25, above, yshift=-0.3em] {
      the $\ell$ independent interactions can be rearranged arbitrarily} ($(rqb8.east) + (0.3em, 0)$);

\node[anchor=west] at (0em, -6em) {\textbf{After sorting state pairs:}};
   \node[anchor=north, minimum width=2\is, minimum height=16em, fill=black!5, align=center] at (8.5\is, -6em) {\ \\[0.5em] \small special treatment\\[-0.4em] \small for collision};

   \foreach \a/\b/\k/\l [count=\i] in {1/1/20/21, 1/1/{2\ell-1}/{2\ell}, 1/2/3/4, 1/2/11/12, 2/1/7/8, d/d, 3/3/9/10, 1/1/5/6, 2/?/2\ell+1/6} {
      \ifthenelse{\equal{\a}{d}}{
         \node at (\i \is, -9em) {$\dots$};
      }{
         \node[draw, circle, inner sep=0, minimum height=1.5em, state\a] (rrqa\i) at (\i \is - 0.22 \is, -9em) {$q_\a$};
         \node[draw, circle, inner sep=0, minimum height=1.5em, state\b] (rrqb\i) at (\i \is + 0.22 \is, -9em) {$q_\b$};
         \node[anchor=south] at (rrqa\i.north) {$s_{i_{\k}}$};
         \node[anchor=south, onif={<\equal{\l}{6}> conflict}] at (rrqb\i.north) {$s_{i_{\l}}$};
      }
   }

   \node[draw, minimum width=0.5\is, minimum height=0.5\is, anchor=north] (trans1) at (8\is, -0.5\is - 9em) {$\delta$};
   \path[draw, red,   semithick, arrows={-\arrowHead}] (trans1.east) -| (rrqb9);
   \path[draw, green, semithick, arrows={-\arrowHead}] (rrqa8) to (trans1);
   \path[draw, green, semithick, arrows={-\arrowHead}] (rrqb8) to (trans1);
   \path[draw, red,   semithick, arrows={-\arrowHead}, bend right] (trans1.east) to (rrqb8);
   \path[draw, red,   semithick, arrows={-\arrowHead}, bend left]  (trans1.west) to (rrqa8);

   \node[anchor=west] at (0em, -15em) {\textbf{After merging interactions with identical state pairs:}};

   \foreach \a/\b [count=\i] in {-/1, -/1, -/2, -/2, -/3, -/3, -/d, 1/1, 2/?} {
      \ifthenelse{\equal{\a}{d}}{
         \node at (\i \is, -17em) {$\dots$};
      }{\ifthenelse{\equal{\a}{-}}{}{
            \node[draw, circle, inner sep=0, minimum height=1.5em, state\a] (gqa\i) at (\i \is - 0.22 \is, -17em) {$q_\a$};
            \node[draw, circle, inner sep=0, minimum height=1.5em, state\b] (gqb\i) at (\i \is + 0.22 \is, -17em) {$q_\b$};
      }}
   }

   \foreach \a/\y in {1/0, 2/2,3/4} {
      \foreach \b [count=\i] in {1, 2, 3} {
         \node[draw, circle, inner sep=0, minimum height=1.5em, state\a, opacity=0.5] (gqa\i\a) at (2* \i \is - 0.22 \is, -17em - \y em) {$q_\a$};
         \node[draw, circle, inner sep=0, minimum height=1.5em, state\b, opacity=0.5] at (2* \i \is + 0.22 \is, -17em - \y em) {$q_\b$};
         \node[anchor=east] at (gqa\i\a.west) {$\mathbf{d_{\a\b}}$ \textcolor{black!50}{$\times$}};
   }}

   \node at (1.5em, -19em) {$D = $ \huge $\Bigg($};
   \node at (28em, -19em) {\huge $\Bigg)$};

   \node[draw, minimum width=0.5\is, minimum height=0.5\is, anchor=north] (trans1) at (8\is, -0.5\is - 17em) {$\delta$};
   \path[draw, red,   semithick, arrows={-\arrowHead}] (trans1.east) -| (gqb9);
   \path[draw, green, semithick, arrows={-\arrowHead}] (gqa8) to (trans1);
   \path[draw, green, semithick, arrows={-\arrowHead}] (gqb8) to (trans1);
   \path[draw, red,   semithick, arrows={-\arrowHead}, bend right] (trans1.east) to (gqb8);
   \path[draw, red,   semithick, arrows={-\arrowHead}, bend left]  (trans1.west) to (gqa8);
\end{tikzpicture} %
}
    \caption{%
        Batch processing uses collision-free runs, long sequences of independent interactions, which can be rearranged and grouped together.%
	}
    \label{fig:batched-interactions}
\end{figure}
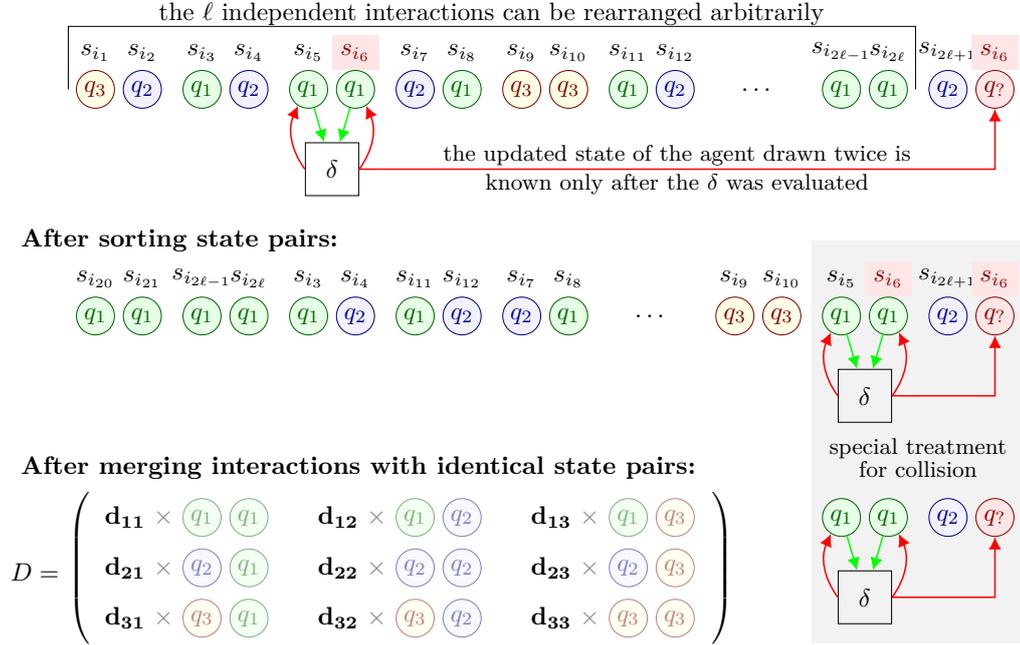

\ifshort
\begin{algorithm}[H]
\else
\begin{algorithm}[t]
\fi
	\caption{\algSimple: The algorithmic framework for simulation in batches.}
	\label{alg:seq-simple}

	\SetKwComment{Comment}{$\triangleright$\ }{}
	\DontPrintSemicolon
	\SetKwInOut{Input}{input\expandafter\?\?}

	\footnotesize

	\Input{~%
		configuration $C$,
		transition function $\delta$,
		number of steps $N$
	}

	$t \gets 0$\;
	\While{t < N}{

	$\ell \gets$ sample length of a collision-free run\;
	\smallskip
	let $D = (d_{ij})$ be a $\q \times \q$ matrix and sample $d_{ij}$ as \Comment*[r]{batch processing} the number of interactions~$(q_i,q_j)$ among $\ell$ interactions\;
	\smallskip
	let $C'$ be an empty configuration\;
	\ForEach{$(q_i, q_j) \in Q^2$}{
		remove from $C$: $d_{ij}$ agents in states $q_i$, and $d_{ij}$ agents in states $q_j$\;
		$(q'_i, q'_j) \gets \delta(q_i, q_j)$\;
		add to $C'$: $d_{ij}$ agents in states $q'_i$, and  $d_{ij}$ agents in states $q'_j$\;
	}

		\eIf(\Comment*[f]{plant a collision}){$\ell$ is even}{
			sample agent $c_1$ without replacement from $C'$\Comment*[r]{collision at $c_1$}
			merge $C'$ into $C$\;
			sample agent $c_2$ without replacement from $C$\;
		}
		{
			sample agent $c_1$ without replacement from $C$\;
			sample agent $c_2$ without replacement from $C'$\Comment*[r]{collision at $c_2$}
			merge $C'$ into $C$\;
		}
		add agents $\delta(c_1,c_2)$ to $C$\;
	$t \gets t + \ell + 1$\;
	}

	\normalsize
\end{algorithm}
 \fi

Recall that \algStepWiseBST has a small memory footprint but a sub-optimal time complexity of $\Theta(N \log {\q})$.
Observe, however, that the underlying binary search tree can update the multiplicity of any existing state in time $\Oh(\log \q)$ \emph{independently} of the changed quantity.
Here, we introduce the new algorithm \algSimple (see \cref{alg:seq-simple} \inshort{in \cref{apx:batch-processing}}) to exploit this observation.
The algorithm uses a binary search tree to store the configuration.
It updates $\Omega(\sqrt n)$ agents in expectation with each access and therefore reduces the time complexity to $\Oh(N (\log n + \q^2) / \sqrt n)$ which is $\oh(N)$ for $\q = \oh(n^{1/4})$ and $N = \Theta(\poly(n))$.

\paragraph{Batching interactions}
In order to coalesce individual updates into batches, \algSimple uses the notion of \emph{collision-free runs} as illustrated in \cref{fig:batched-interactions} \inshort{in \cref{apx:batch-processing}}.
We interpret the execution of a protocol as a sequence $i_1, i_2, \ldots$ where at time $t$ agents $i_{2t-1}$ and $i_{2t}$ interact.
Let $\ell$ be the largest index such that all $i_1, \ldots, i_\ell$ are distinct.
Then, the first $\ellHalf$ interactions are independent of each other and can be rearranged in any order.
We refer to them as a collision-free run of length~$\ell$.
If $\ell$ is odd, the first agent of the $(\ellHalf+1)$-th interaction is also considered collision-free.
Since we are free to reorder the interactions, we can group all interactions of states $(q_i, q_j)$ together, evaluate~$\delta(q_i, q_j)$, and update all accordingly affected states in one step.

Now instead of sampling a sequence of agents and partitioning the sequence into collision-free runs, we take the opposite direction.
We first sample only the length~$\ell$ of a collision-free run from the appropriate probability distribution (see below).
Then, we randomly match $\ell$ agents as discussed below.
Finally, we reuse one of the agents from the matching in order to plant a collision.
These steps are repeated until at least $N$ interactions are simulated.

\paragraph{Matching Agents}
We simulate sampling $\ell$ agents without replacement to construct a collision-free run of length~$\ell$.
While we cannot afford to draw the agents individually, we only need to know how many interactions~$n_{ij}$ of each state pair $(q_i, q_j)$ we encountered.
Thus, a run can be modeled by a $\q {\times} \q$ matrix $D = (n_{ij})$ with $\sum_{ij} n_{ij} = \ellHalf$. (If $\ell$ is odd, we remove one agent and treat it individually.)

To obtain $D$, we first sample the row sums $D_i = \sum_j n_{ij}$ of the matrix from a multivariate hypergeometric distribution.
This simulates sampling $\ellHalf$ initiating agents without replacement.
We then sample values within each row analogously to find the matching responding agents.
Sampling $D$ takes $\Oh(\q^2)$ time in total since each individual sample from a hypergeometric distribution can be computed in $\Oh(1)$ time \cite{DBLP:conf/wsc/Stadlober89}.

For correctness, note that our sampling approach corresponds to first selecting \ellHalf agents as initiators and and then \ellHalf agents as responders.
That is, we first sample agents $i_1, i_3, \dots, i_{2\ellHalf - 1}$ and then agents $i_2, i_4, \dots, i_{2\ellHalf}$ (instead of the natural interleaved variant $i_1, i_2, \dots, i_{2\ellHalf}$).
Since, each draw is taken uniformly at random, the permutation does not change the distribution (see \cref{app:reordering-argument} for a formal proof).

\paragraph{Length of a Collision-Free Run}
In the following, we analyze the length~$\ell$ of a collision-free run.
Observe that the following analysis is similar to the analysis of a generalized variant of the birthday problem \cite{DBLP:conf/sacrypt/KuhnS01}.
We consider a generalization which we also use in \cref{sec:merging-runs}.
We assume that $r$ agents have already interacted and ask how many more collision-free agents can be added.
Formally we define the distribution $\colDistr{n}{r}$ as follows.

\begin{definition}
Consider a sequence $a_1, a_2, \ldots$ of agents sampled independently and uniformly at random.
Let $A_0$ be a set of $r$ initially prescribed agents and let $A_i = A_{i-1} \cup \{a_i\}$ be the set of agents after $i$ draws.
We define the random variable $\ell$ as the smallest index s.t.\ $a_\ell \in A_{\ell - 1}$.
We say $\ell \distrAs \colDistr{n}{r}$, where $n$ is the total number of agents and $r$ is the number of prescribed agents.
\end{definition}

\ifshort
In \cref{sec:sampling-length-collision-free-run} we discuss how we can sample from this distribution using the inverse sampling technique~\cite{DBLP:books/sp/Devroye86}.
We find that sampling~$\ell$ takes $\Oh(\log n)$ time.
In practice, this is comparable to the time it takes to sample a hypergeometric random variate. In the following, we show basic properties of \colDistr{n}{r} in order to show bounds on the runtime of \algSimple.
\fi

\begin{lemma}[name=,restate=restateLemCollisionDistance,label=lem:collision-distance]
Let $\ell \distrAs \colDistr{n}{r}$.
\ifshort
	Then $\ell$ has support $\set{1, \dots, n - r}$ and distribution\par
$\displaystyle \prob{\ell = k} =  n^{- (k + 1)} (n - r)! (r + k) / (n - r - k)! $.
\else
	Then $\ell$ has distribution
	\begin{equation*}
		\prob{\ell = k} = \begin{cases}
		 \displaystyle n^{- (k + 1)} \frac{(n - r)!}{(n - r - k)!} (r + k) & \text{if } 0 < k \le n - r \\
		 0 & \text{otherwise.}
		\end{cases}
	\end{equation*}
\fi
\end{lemma}
\def\proofCollisionDistance{%
	\begin{proof}
		Consider an urn with $n$ marbles.
		Initially, $r$ marbles are red, while the remaining $n - r$ marbles are green.
		We now take out one marble at a time:
		if it is green, we keep on going (think of a traffic light) and put a red one back in.
		If we take a red marble, we stop.
		Observe that the number of marbles we take out is exactly $\ell$ as above, as the green marbles represent new unconsidered agents while the red ones represent agents in $A_{\ell - 1}$.
		This directly leads to the acclaimed distribution:
		\begin{equation*}
		\prob{\ell = k}\quad = \quad
			\underbrace{\prod_{i=0}^{k - 1} \frac{(n - r) - i}{n}}_{\text{select $k$ out of $n - r$}}
			\quad \cdot \quad
			\underbrace{\frac{r + k}{n}}_\text{$(k{+}1)$-th is red} \qedhere
		\end{equation*}
	\end{proof}
}
\iflong\proofCollisionDistance\fi

\inlong{\goodbreak}

\begin{lemma}[name=,restate=restateLemExpectedBatchLength,label=lem:expected-batch-length]
	Let $\ell \sim \colDistr{n}{0}$. Then $\Ex{\ell} = \Theta(\sqrt{n})$.
\end{lemma}
\def\proofExpectedBatchLength{%
	\begin{proof}
		We first upper bound $\Ex{\ell} = \Oh(\sqrt{n})$ and then give a matching lower bound $\Ex{\ell} = \Omega(\sqrt{n})$.
		In both cases, we write $\Ex{\ell} = \sum_{i=0}^{n}\prob{\ell \ge i}$ and split the sum at $\sqrt n$.
		Then we bound both terms appropriately.
		Observe that for some fixed value $i$ we have $\Prob{\ell \geq i} = \prod_{j = 0}^{i - 1}\left(1 - {j}/{n}\right)$.
		For the upper bound on $\Ex{\ell}$ we get
		\begin{align*}
			\Ex{\ell}
			&= \sum_{i=0}^{n}\prob{\ell \ge i}
			 = \sum_{i=0}^{n}\prod_{j = 0}^{i - 1}\left(1 - \frac{j}{n}\right)
			 \le \sum_{i=0}^{\sqrt n - 1} 1 + \sum_{i = \sqrt n}^{\infty}\left(1 {-} \frac{\sqrt n}{n}\right)^i
			 \leq 2 \sqrt n .
		\intertext{Similarly, we get for the lower bound on $\Ex{\ell}$ that}
		  \Ex{\ell}
		  &=  \sum_{i=0}^{n}\prob{\ell \ge i}
		   =  \sum_{i = 0}^{n}\prod_{j=0}^{i-1}\left(1 - \frac{j}{n}\right)
		  \ge \sum_{i = 0}^{\sqrt{n}}\prod_{j=0}^{i-1}\left(1 - \frac{\sqrt n}{n}\right)
		   =  \sum_{i = 0}^{\sqrt{n}}\left(1 - \frac{1}{\sqrt n}\right)^i
\\	  &=  \sqrt n \left( 1 - \left(1 - \frac{1}{\sqrt n}\right)^{\sqrt n + 1} \right)
      \ge \sqrt n \left( 1 - e^{-1} \right) .
		\end{align*}
Therefore we have $\Ex{\ell} = \Theta(\sqrt{n})$.
	\end{proof}
}
\iflong\proofExpectedBatchLength\fi

\noindent Using \cref{lem:expected-batch-length}, we are now ready to bound the runtime and space complexity of \algSimple.

\begin{theorem} Let $n$ be the number of agents and $\q$ the number of states.
	\algSimple simulates $N$ interactions in $\Oh(N (\q^2  + \log n ) / \sqrt{n})$ expected time using $\Theta(\q \log{n})$ bits.
\end{theorem}

\begin{proof}
	According to \cref{lem:expected-batch-length}, each batch simulates $\Theta(\sqrt{n})$ interactions in expectation.
	It takes $\Oh(\log n)$ time to sample the length of a collision-free run~$\ell$ (see \cref{sec:sampling-length-collision-free-run})  and $\Oh(\q^2)$ time (\cf \cite{DBLP:conf/wsc/Stadlober89}) to sample the interaction numbers and process the interactions for all pairs of states.
	This implies the runtime complexity.
	The space complexity follows immediately from the binary search tree used to store the configuration.
\end{proof}
\section{Merging Batches}
\label{sec:merging-runs}

In an empirical evaluation, we found that our implementation of algorithm~\algSimple spends most time in the batch processing step (to sample and transition the $\q \times \q$ matrix~$D$);
this is especially true for complex protocols with non-trivial state space sizes.
As the matrix sampling cost is independent of the length~$\ell$ of the underlying collision-free run, we modify the algorithm to support more than one collision per batch processing step.

\paragraph{Introducing Epochs}
An execution of the improved algorithm~\algMulti logically consists of several epochs.
For each epoch, the algorithm samples the \emph{lengths} $\ell_1$, $\ell_2$, \dots{}, $\ell_\rho$ of multiple collision-free runs $R_1, \dots, R_\rho$.
As no agent may appear twice in the union of those collision-free sequences, later runs become shorter in expectation ($\Ex{\ell_{i+1}} < \Ex{\ell_{i}}$), naturally limiting the number~$\rho$ of runs per epoch.
After each run $R_i$, we plant one collision, \ie an interaction with an agent that was already considered in the current epoch.
An epoch concludes with a single batch processing step, in which matrix~$D$ is sampled and processed analogously to algorithm~\algSimple.

\paragraph{Tracking Dependencies}
While algorithm~\algSimple only reorders and groups together independent interactions, our improved algorithm~\algMulti{} delays most interactions until the end of the epoch.
To do so, the algorithm conceptually assigns each agent one of three types, and updates these labels as it progresses through the epoch:
\begin{itemize}
	\item \untouched agents did not interact in the current epoch.
	Hence, all agents are labeled \untouched at the beginning of an epoch.

	\item \updated agents took part in at least one interaction that was already evaluated.
	Thus, \updated agents are already assigned their most recent state.

	\item \delayed agents took part in \emph{exactly one} interaction that was not yet evaluated.
	Thus, \delayed agents are still in the same state they had at the beginning of the epoch, but are scheduled to interact at a later point in time.
	We additionally require that their interaction partner is also labeled \delayed.
\end{itemize}

Analogously to algorithm~\algSimple, we maintain two urns~$C$ and $C'$.
Urn $C'$ contains \updated agents, while urn~$C$ stores \untouched and \delayed agents (or in other words, all agents whose state was not updated in the current epoch).
At any point in time, an agent is either in $C$ or $C'$ meaning that $|C| + |C'| = n$.
Due to symmetry, we do not explicitly differentiate \untouched from \delayed agents.
We rather maintain only the number~$T$ of \delayed agents and lazily select them while planting collisions or during batch processing.

If a \delayed agent~$a$ is selected while planting a collision, it takes part in a second interaction and ---by definition--- cannot be labeled \delayed any more.
Thus, we randomly draw a second \delayed agent~$b$, evaluate their transition, store the updated state of~$b$ in~$C'$, and directly evaluate~$a$ again in the planted collision.
Finally, we decrease $T \gets T-2$ as agents $a$ and $b$ changed their labels from \delayed to \updated.
Observe that we might repeat this step in the (unlikely) case that a planted collision involved two formerly \delayed agents.

\paragraph{Length of an Epoch}
We now analyze the length of an epoch.
We start by extending the analysis of $\colDistr{n}{r}$ to the $r = \Omega(\sqrt n)$ regime (reached after $\Oh(1)$ runs \whp).
The following lemmas establish expected value and concentration. \inshort{See \cref{sec:length-of-a-collision-free-run} for the proofs.}

\begin{lemma}[restate=restateExpectedBatchLengthPrescribed,label=lem:expected-batch-length-prescribed]
	Let $\ell \sim \colDistr{n}{r}$ and $r = \Omega(\sqrt{n})$. Then $\Ex{\ell} = \Theta(n/r)$.
\end{lemma}
\def\proofExpectedBatchLengthPrescribed{%
	\begin{proof}
		The proof follows analogously to \cref{lem:expected-batch-length}.
		Again, we start with the upper bound.
		\begin{align*}
		\Ex{\ell} &= \sum_{i=0}^{n-r}\prob{\ell \ge i} = \sum_{i = 0}^{n - r}\prod_{j = 0}^{i - 1}\left(1 - \frac{j + r}{n}\right) \le \sum_{i = 0}^{\infty}\left(1 - \frac{r}{n}\right)^i = \frac{n}{r}.
		\end{align*}
		For the lower bound we derive a general result for arbitrary $r$.
		\begin{align*}
		\Ex{\ell} &= \sum_{i=0}^{n-r}\prob{\ell \ge i} = \sum_{i = 0}^{n - r}\prod_{j = 0}^{i - 1}\left(1 - \frac{j + r}{n}\right) \ge \sum_{i = 0}^{r-1}\prod_{j = 0}^{i - 1}\left(1 - \frac{j + r}{n}\right) \ge \sum_{i=0}^{r-1}\left(1 - \frac{2r}{n}\right)^i \\
		&= \frac{n}{2r} \left(1 - \left(1 - \frac{2r}{n}\right)^{r}\right) \ge \frac{n}{2r}(1 - e^{-2r^2/n}).
		\end{align*}
		The last inequality holds since $e^{-2r^2/n}$ constitutes an upper bound for $(1 - 2r/n)^r$ as it can be rewritten as $(1 - 2r/n)^{n\cdot r/n}$ and $(1 - 2r/n)^n \le e^{-2r}$.
		For $r = \Omega(\sqrt{n})$ the second factor $(1 - \exp(-2r^2/n))$ is $\Omega(1)$ which proves the claim.
	\end{proof}
}
\iflong\proofExpectedBatchLengthPrescribed\fi

\inlong{\goodbreak}

\begin{lemma}[restate=restateBatchLengthProb,label=lem:batch-length-prob]
	Let $\ell \sim \colDistr{n}{r}$ and $r = \Omega(\sqrt{n})$. Then $\ell = \Theta(n/r)$ with probability $1 - o(1)$.
\end{lemma}
\def\proofBatchLengthProb{
	\begin{proof}
		We prove the claim by showing that $\prob{\ell < t}$ and $\prob{\ell > t}$ are $o(1)$ for $t = o(n/r)$ and $t = \omega(n/r)$, respectively.
		\begin{align*}
		\prob{\ell < t} &= 1 - \prob{\ell \ge t} = 1 - \prod_{i=0}^{t - 1}\left(1 - \frac{i + r}{n}\right) \le 1 - \left(1 - \sum_{i=0}^{t - 1}\frac{i + r}{n} \right)\\
		&= \frac{2tr + t(t-1)}{2n} \le \frac{2tr + t^2}{2n}.
		\end{align*}
		Applying the Weierstrass product inequality yields the first inequality.
		Further, with $r = \Omega(\sqrt{n})$ we have $n/r = O(\sqrt{n})$ and thus $t = o(n/r)$ such that $t^2 = o(n)$ and $tr = o(n)$.
		\begin{align*}
		\prob{\ell > t} & = \prod_{i=0}^{t}\left(1 - \frac{i+r}{n}\right) \le \left(1 - \frac{r}{n}\right)^t \le e^{-rt/n}.
		\end{align*}
		If $t = \omega(n/r)$ then $rt/n = \omega(1)$ and the claim follows.
	\end{proof}
}
\iflong\proofBatchLengthProb\fi

Intuitively, \cref{lem:expected-batch-length-prescribed} shows that for sufficiently many prescribed agents~$r$, the probability of drawing a colliding agent remains approximately $r/n$ throughout the run.
Similar to a geometric distribution, this results in a concentrated expected length of $\Theta(n/r)$.
We now estimate the number of agents sampled after $\rho$ runs.

\begin{lemma}[restate=restateMultiBatchNumberInteractions,label=lem:multi-batch-number-of-interactions]
	Let $L_k = \sum_{i=1}^{k}\ell_i$ be the number of agents drawn in an epoch with $k$ runs.
	Then, for $\q = \omega(\sqrt{\log n})$, $\q = o(\sqrt{n\log n})$ and $\rho = \Oh(\q^2/\log n)$ we have $\Ex{L_\rho} = \Theta(\sqrt{\rho n})$.
\end{lemma}
\def\proofMultiBatchNumberInteractions{%
	\begin{proof}
		The variable $L_k$ equivalently corresponds to the number of marbles $B(k, n)$ that need to be drawn in the birthday problem s.t.~$k$ coincidences occur.
		The asymptotics of $\Ex{B(k, n)}$ have first been studied by \textcite{DBLP:conf/sacrypt/KuhnS01} for the cases that $k = o(n^{1/4})$.
		Their results have since been improved by \textcite{DBLP:journals/rsa/ArratiaGK16} where the asymptotic bounds on the moments of $B(k, n)$ have been calculated for more general conditions on $k$.
		By \cite[Corollary 12]{DBLP:journals/rsa/ArratiaGK16} for $k$ a function of $n$, \ie $k = k_n$ where $k_n \to \infty$ and $k_n/n \to 0$ it holds that
		\[ \Ex{B(k_n, n)} \sim \sqrt{2nk_n} \quad \text{as } n \to \infty. \]
		By assumption the conditions are met since $\rho = \omega(1)$ and $\rho = o(n)$, thus $\Ex{L_\rho} = \Theta(\q\sqrt{n/\log n}) = \Theta(\sqrt{\rho n})$.
	\end{proof}
}
\iflong\proofMultiBatchNumberInteractions\fi

\paragraph{Complexity}
In order to analyze \algMulti's runtime, we first establish the time required per epoch, and then bound the total expected runtime and memory requirements.

\begin{lemma}[restate=restateAlgMultiRunningTime,label=lem:alg-multiple-running-time]
	\algMulti takes time $\Oh(\rho \log n+ \q^2)$ for an epoch of $\rho$ runs.
\end{lemma}
\def\proofAlgMultiRunningTime{%
	\begin{proof}
		Planting a collision is done by drawing the two interacting agents from the appropriate urns and setting them to be \updated which requires $\Oh(1)$ time.
		Sampling the length of a single collision-free run takes time $\Oh(\log n)$ (see \cref{sec:sampling-length-collision-free-run}).
		For the final batch-processing step \algMulti takes $\Theta(\q^2)$ time independently of the number of \delayed agents.
	\end{proof}
}
\iflong\proofAlgMultiRunningTime\fi

\inlong{\goodbreak}

\begin{theorem} Let $n$ be the number of agents and $\q$ the number of states.
    \algMulti simulates $N$ interactions in $\Oh(N \q / \sqrt{n / \log n})$ expected time if $\omega(\sqrt{\log n}) \leq \q \leq \oh(\sqrt{n \log n})$.
\end{theorem}

\begin{proof}
	Combining \cref{lem:multi-batch-number-of-interactions, lem:alg-multiple-running-time}, we find a runtime of
	$\Oh(N(\rho \log n + \q^2) / \sqrt{\rho n})$.
	Setting $\rho = \Theta(\q^2 / \log n)$ balances the cost of sampling runs and planting collisions with the cost of batch processing, and thus does not increase the asymptotic cost per epoch.
	Higher values of $\rho$ only increase the expected time complexity.
\end{proof}

\noindent \algMulti has sub-constant work per interaction for $\q {=} \oh\left(\sqrt{\frac{n}{\log n}}\right)$ and $N {=} \Theta(\poly n)$.
\section{Heuristics and Implementation Details}
\label{sec:heuristics}
Implementations of all discussed simulators (including scripts to reproduce figures and numbers included in this paper) are freely available.
\iflong\footnote{\url{https://projects.algorithm.engineering/population-protocols}}\fi
In the following, we highlight important aspects necessary to obtain simulators that are both fast in practice and highly customizable.

\iflong
\begingroup
\providecommand\nolinenumbers{}
\nolinenumbers
\begin{lstlisting}[
	caption={Example of a simple \problem{Leader Election} protocol.},
	label=lst:leader-election,
]
struct LeaderElectionProtocol : public
  Protocols::DeterministicProtocol, Protocols::OneWayProtocol
{
  enum Roles : state_t {Follower = 0, Leader = 1};

  state_t operator() (state_t initiator, const state_t responder) const {
    if (initiator == responder) initiator = Follower;
    return initiator; // simplification as one-way protocol
  }

  state_t num_states() const {return 2;}
};
\end{lstlisting}
\begin{lstlisting}[
	caption={%
		A non-deterministic protocol in which either the initiator or the responder circularly increments its state by one. In each interaction the active one is determined by tossing a fair coin.
	},
	label=lst:non-deterministic,
]
state_t operator() (state_t initiator, state_t responder,
                    count_t num, auto& assign) const {
  count_t n = std::binomial_distribution<>{num, 0.5}(gen_);

  //    (state                       , number of agents);
  assign(initiator + 1 %
  assign(responder                   , n);

  assign(initiator                   , num - n); // responder advances
  assign(responder + 1 %
}
\end{lstlisting}
\endgroup
 \fi

All simulators are implemented in C++ and use compile-time specializations to implement specific protocols and experimental setups.
See \cref{lst:leader-election} \inshort{in \cref{apx:implementation-details}} for a minimal example of such a protocol.
It specifies only the number~$\q$ of states and the transition function~$\delta$.

In contrast to pure deterministic functions, non-deterministic transition functions (possibly with side-effects) have to be informed about every interaction carried out.\footnote{
	Observe that many protocols with non-deterministic transition functions have been derandomized, see, \eg the notion of (biased) synthetic coins in \cite{DBLP:conf/soda/AlistarhAEGR17,DBLP:conf/soda/BerenbrinkKKO18}.
	While this is supported by the simulator, we feel it is more convenient to offer the most expressive interface possible.
}
To allow for batch processing, we cannot use the natural invocation order.
Instead, we inform the protocol how often a state pair will interact within an epoch.
It is then expected to assign all participating agents to the appropriate states (see \cref{lst:non-deterministic} \inshort{in \cref{apx:implementation-details}} for an example).

\def\IST{\textsc{Ist}\xspace}
\def\cdf{\text{cdf}}
\newcommand{\colCDF}[3]{\ensuremath{F_{#1,#2}(#3)}}

\iflong
\subsection{Sampling the Length of a Collision-Free Run}
\label{sec:sampling-length-collision-free-run}
Recall that \algSimple and \algMulti repeatedly sample the length~$\ell$ of a collision-free run.
In the following we discuss how this sampling can be implemented using the inverse sampling technique (\IST)~\cite{DBLP:books/sp/Devroye86}:
let $\cdf(x)$ be the cumulative density function of a target distribution.
Then, \IST draws a uniform variate~$U$ from $[0; 1]$, solves $U = \cdf(x)$ for $x$, and returns it as the sample.
We denote the CDF of $\colDistr{n}{k}$ as $\colCDF{n}{k}{t}$.
\cref{lem:collision-distance} yields:
\begin{equation*}
	1 - \colCDF{n}{k}{t}
	  = \prob{\ell {>} t}
	  = \prod_i^{t} \frac{n {-} k {-} i}{n}
	  = \frac{1}{n^t} \frac{ (n-k)! }{ (n-k-t-1)! }
	  \stackrel{(x-1)! = \Gamma(x)}{=} n^{-t} \frac{\Gamma(n - k + 1)}{\Gamma(n-k-t)}.
\end{equation*}
Since we are not aware of an inverse that can be evaluated fast, we numerically solve $U = \colCDF{n}{k}{t}$ for $t$.
To avoid numerical instabilities, we rewrite the expression in terms of $\log\Gamma(x)$, which is available as the C standard function \texttt{lgamma}$(x)$:
\begin{equation*}
	U = 1 - n^{-t} \frac{\Gamma(n - k + 1)}{\Gamma(n-k-t)}
	\quad \Leftrightarrow\quad
	\log\left(1 - U\right) = \log\Gamma\left(n - k + 1\right) - \log\Gamma\left(n-k-t\right) - t \log n
\end{equation*}

Lacking a cheap derivative of the RHS, we rely on first-order numerical inversion methods only.
In this context, an ad-hoc combination of binary search and regula-falsi gave most consistent results.
We jump-start the search using a small look-up table containing lower and upper bounds on~$t$ for intervals of $U$ and $k$.

While the method requires $\Oh(\log n)$ evaluations of $\colCDF{n}{k}{\cdot}$, we observe less than ten calls on average for $n = 2^{50}$.
The resulting sampling algorithm has a practical runtime comparable to the sampling of hypergeometric random variates.
Since the latter is sampled much more frequently, further optimizations will yield limited results to the total runtimes of $\algSimple$ and $\algMulti$.
 
\subsection{Heuristics}

\else
\fi

The frequent sampling of hypergeometric random variates, dominates \algMulti's runtime.
In the following, we discuss three heuristics to reduce this number.

\paragraph{Renaming}\label{subsec:heuristic-renaming}
In the \emph{renaming heuristic} we exploit the observation that agents are typically not uniformly distributed over all states.
Instead there are often sparsely populated states which are seldom hit when sampling an agent.

It can be beneficial to consider these states last.
\algStepWiseLinear's linear search, for instance, stops as soon as the sampled state is found.
The same is true when sampling \algSimple's and \algMulti's interaction matrices:
for row~$i$, we draw $D_i$ agents from a multivariate hypergeometric distribution.
This is implemented by obtaining $\q{-}1$ properly parametrized hypergeometric variates;
the process terminates early once all $D_i$ agents have been sampled.

In both examples, we maximize the probability of early stopping by considering highly populated states first.
To this end, we maintain a permutation $\pi\colon [\q] \to [\q]$ that sorts states decreasingly by their sizes.
We then process states in the order indicated by~$\pi$.
If this permutation is updated once every $\Omega(\q \log \q)$ interactions, the sorting step becomes asymptotically negligible for sequential simulators.
For \algSimple and \algMulti, $\pi$ can be updated once every $\Omega(\log q)$ batches.

\paragraph{Partitioning}\label{subsec:heuristic-partitioning}
\begin{figure}
	\centering
	\scalebox{0.8} {%
\begin{tikzpicture}
	\node (matrix) {
		\bgroup
		\begin{tabular}{c | c c c c}
		  & 0 & 1 & 2 & 3  \\
		\hline
		0 & 0 & 1 & 0 & 0 \\
		1 & 1 & 1 & 2 & 1 \\
		2 & 2 & 2 & 2 & 3 \\
		3 & 0 & 3 & 3 & 3
		\end{tabular}
		\egroup
	};

 \node (singlerow) at (13.7em, 0em) {
	 \bgroup
	 \begin{tabular}{c c c c}
	 & & &   \\
	 & & & \\
	 1 & 1 & 2 & 1 \\
	 & & & \\
	 & & &
	 \end{tabular}
	 \egroup
	};

	\node[anchor=north] at (0em, 5em) {\textbf{Full matrix}};
	\node[anchor=north] at (13.7em, 5em) {\textbf{Row for $q_u = 1$}};
	\filldraw[opacity=0.2,color=red] (13.7em, 0.5em) -- (10.6em, 0.5em) -- (10.6em, -0.5em) -- (13.7em, -0.5em) -- cycle;
	\filldraw[opacity=0.2,color=red] (16.9em, 0.5em) -- (15.7em, 0.5em) -- (15.7em, -0.5em) -- (16.9em, -0.5em) -- cycle;
	\filldraw[opacity=0.2,color=blue] (15.5em, 0.5em) -- (13.9em, 0.5em) -- (13.9em, -0.5em) -- (15.5em, -0.5em) -- cycle;

	\filldraw[opacity=0.2,color=red] (0.8em, 0.5em) -- (-2.3em, 0.5em) -- (-2.3em, -0.5em) -- (0.8em, -0.5em) -- cycle;
	\filldraw[opacity=0.2,color=red] (4em, 0.5em) -- (2.8em, 0.5em) -- (2.8em, -0.5em) -- (4em, -0.5em) -- cycle;
	\filldraw[opacity=0.2,color=blue] (2.6em, 0.5em) -- (1em, 0.5em) -- (1em, -0.5em) -- (2.6em, -0.5em) -- cycle;

	\node[anchor=north] at (10em, -2em) {Group ($q'_u = 1$)};
	\node[anchor=north] at (18em, -2em) {Group ($q'_u = 2$)};
	\draw[->,semithick] (14.7em, -0.8em) -- (17em, -1.8em);
	\draw[->,semithick] (12.1em, -0.8em) -- (11em, -1.8em);
	\draw[->,semithick] (16.3em, -0.8em) -- (12em, -1.8em);
\end{tikzpicture} %
}
	\vspace{-0.5em}
	\caption{%
		Simplified transition matrix~$\Delta'$ for a clock with period $m {=} 4$.
		The initiator's phase (row) is circularly incremented only when matched with a suitable responder (column).
	}
	\label{fig:phase_clock_transition_submatrix}
\end{figure}

If $\delta$ is a deterministic function, we can model it as a matrix~$\Delta \in (Q \times Q)^{\q \times \q}$.
The matrix of a deterministic one-way protocol can be further simplified to $\Delta' \in Q^{\q \times \q}$ since the states of the responders remain unchanged.

For many meaningful protocols, the entries of $\Delta'$ are not random but exhibit some structure.
As an example, consider the simplified\footnote{%
	Formally the states of the phase-clock are pairs $(x, b)$ where $x$ represents the \emph{phase}, and $b$ marks an agent as \emph{leader}.
	For the sake of simplicity we assume that all agents are followers.
} phase-clock transition matrix illustrated in \cref{fig:phase_clock_transition_submatrix}.
Here, each row contains only two different output states.
The \emph{partitioning heuristic} uses this observation during the batch steps of \algSimple and \algMulti.
When sampling $D_i$ responders for initiators in state~$q_i$, we group together all entries in the $i$-th row of $\Delta'$ that assign the same new state to the initiating agents.
It then suffices to draw one random hypergeometric variate per group.

Note that the heuristic does not reduce the runtime complexity of the algorithm --- even if we precompute the partitioning.
This is due the fact that we still need to compute the population sizes for each group which involves $\Theta(\q)$ additions per row.
For pathological protocols, the number of hypergeometric random variates required for each row remains $\Omega(\q)$.
A simple worst-case protocol is the transition function $\delta(q_u, q_v) = (q_v, q_v)$.

\paragraph{Skipping}
Generalizing \emph{partitioning heuristic} to two-way protocols tends to be ineffective in practice:
since initiator and responder may both update their states, the transition matrix is often more fragmented.
The partitioning overhead then easily exceeds the potential savings.
\inlong{\par}
For such protocols \algMulti uses a coarser partitioning, and only detects and skips transitions that preserve the configuration (\ie $\delta(q_u, q_v) = (q_u, q_v)$ or $\delta(q_u, q_v) = (q_v, q_u)$).
\iflong
To see how this applies, let $\delta(q_u, q_v)$ and $\delta(q_u, q_w)$ be configuration preserving transitions.
Consider the sampled row sum $D_u$ for $q_u$ where we in a second step sampled the interaction counts $d_{uv}$ and $d_{uw}$ of interactions of agents in state $q_u$ with agents in states $q_v$ or $q_w$.
The batch-processing algorithm would simply take both counts $d_{uv}$ and $d_{uw}$ and update the urn $C'$ by subsequent additions.
In this concrete scenario we group $q_u$ with $q_v$ and only sample one count rather than two reducing the number of random variates.
\fi

\subsection{Dynamic Epoch Lengths}
Recall that $\algMulti$ is split into several epochs that each consist of multiple collision-free runs.
In our implementation, we add runs to an epoch until the number of interactions exceeds a specified threshold $T$.
The value of $T$ has to be chosen as a trade-off between the cost of adding another run (\ie sampling the run length and planting a collision) versus the diminishing return it yields (as later runs become shorter in expectation).
This trade-off depends on the protocol and its configuration.
For instance, the batch processing cost of a convergent protocol may become smaller compared to the initial costs (\eg when most agents are in only a small fraction of the states).

As the trade-off is dynamic, we maximize the throughput using a control loop that dynamically optimizes the length of an epoch.
Given the currently best value known for $T$, it increases (and later decreases) $T$ to $1.1T$ and $0.9T$, respectively.
For each of the three values, we measure the throughput, chose the $T$ which maximizes it and repeat.
Since the throughput response curve is single-peaked, the process will find a nearly optimal $T$.
\section{Experimental Evaluation}
\label{sec:simulations}
In the following, we empirically evaluate the various simulation algorithms.
The code is compiled using \texttt{g++-8.3} with flags \texttt{-O3 -march=native} and executed on the following system:
 $2\times$ Intel Xeon Gold 6148 CPU @ \SI{2.4}{\GHz} (40~cores/80 hardware threads in total),
 \SI{192}{\gibi\byte} DDR4 RAM @ \SI{2666}{\MHz}.
Each data point is the median of at least five measurements (using different random seeds); error bars indicate their the standard deviation.

\algStepWiseBST's search tree is implemented as an array with breath-first-indexing (\ie the weight of node~$i \ge 1$ is stored at $A[i]$;
its left child is at index $2i$, its right child at $2i {+} 1$).
The implementation uses predicate logic to reduce pipeline stalls due to conditional branching.
\algStepWiseArray uses an array with \SI{32}{\bit} words to store states.
We additionally consider \algStepWiseArrayPre which prefetches states for eight\footnote{%
    This is the optimal value measured for this CPU type and slightly varies between machines.
} interactions ahead of time as a latency hiding technique.

\begin{figure}
    \begin{center}
        \includegraphics[width=0.9\textwidth]{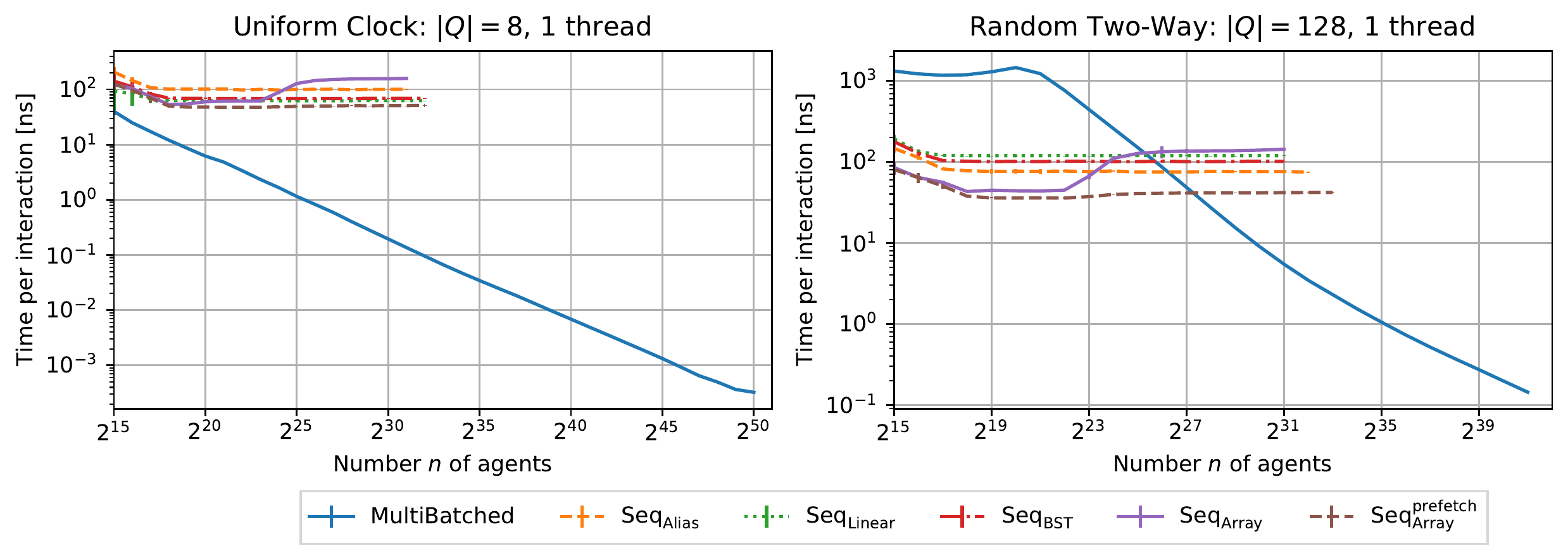}
    \end{center}

    \vspace{-1em}

    \caption{
        Processing time per interaction as function of the number of agents $n$.
        Each series ends with the largest value $n$ for which the median of the total processing time is below \SI{400}{\second}.
    }
    \label{fig:benchmark_selection_batch}
\end{figure}

The runtime of most simulators has non-trivial dependencies on the input parameters, protocol, and state distribution.
Hence, we simulate a small number $N = n$ of interactions to prevent measuring artifacts caused by significant changes in the state distributions.
\algMulti typically simulates slightly more interactions due to the batching granularity.
Since the runtime of all simulators is linear in~$N$, we always report the time \emph{per interaction} to ease extrapolation.

\noindent Three different protocols are used to highlight certain aspects of the algorithms.
\begin{itemize}
    \item \emph{Uniform Clock} and \emph{Running Clock} implement the same deterministic one-way protocol phase-clock protocol inspired by~\cite{DBLP:conf/soda/GasieniecS18}.
    In the \emph{running} variant, all agents start in the first phase, and $\sqrt{n}$ of them are marked.
    Due to the choice of parameters, we expect only one out of $\Theta(\sqrt n)$ interactions to change states.
    Thus, even at the end of the benchmark, the population is still highly concentrated in the lowest phase.
    The \emph{uniform} variant, in contrast, evenly distributes marked and unmarked agents over all phases.
    This results in a constant update probability per interaction.

    \item The \emph{Random Two-Way} protocol uses a deterministic transition function $\delta(q_i, q_j) = d_{ij}$ where each $d_{ij}$ is initially drawn independently and uniformly at random from $Q$.
    Initially agents are evenly distributed over all states.
\end{itemize}

\paragraph{Number of Agents}
We begin our experimental study by investigating the dependencies on the problem size~$n$.
To this end, we search for the largest number~$n$ of agents that a simulator can simulate within a fixed time budget of \SI{400}{\second}.
\Cref{fig:benchmark_selection_batch} reports such measurements for two different settings (see \cref{fig:benchmark_grid_with_batch} in \cref{apx:additional-benchmarks} for the full set).

For the \emph{Uniform Clock} protocol with $\q = 8$ states, the fastest \algStepWise variant reaches $n=2^{32}$ within the time budget.
In the same time, \algMulti simulates $N=n=2^{50}$ interactions.
For \emph{Random Two-Way}, the ratio between the achievable population sizes is smaller but still exceeds three orders of magnitude.
We attribute the different ratios mainly to the batching step.
\algMulti requires $\Theta(\q^2)$ hypergeometric variates per batching step for the latter protocol, since neither the partitioning nor the skipping heuristics (see \cref{sec:heuristics}) are effective on a featureless transition matrix with uniformly random states.
For the \emph{Uniform Clock} protocol, the partitioning heuristic reduces the number of hypergeometric random variates to less than $2\q$ per batch (\cf \cref{fig:benchmark_grid_with_batch} in \cref{apx:additional-benchmarks}).

\paragraph{Number of States}
\begin{figure}
    \includegraphics[width=\textwidth]{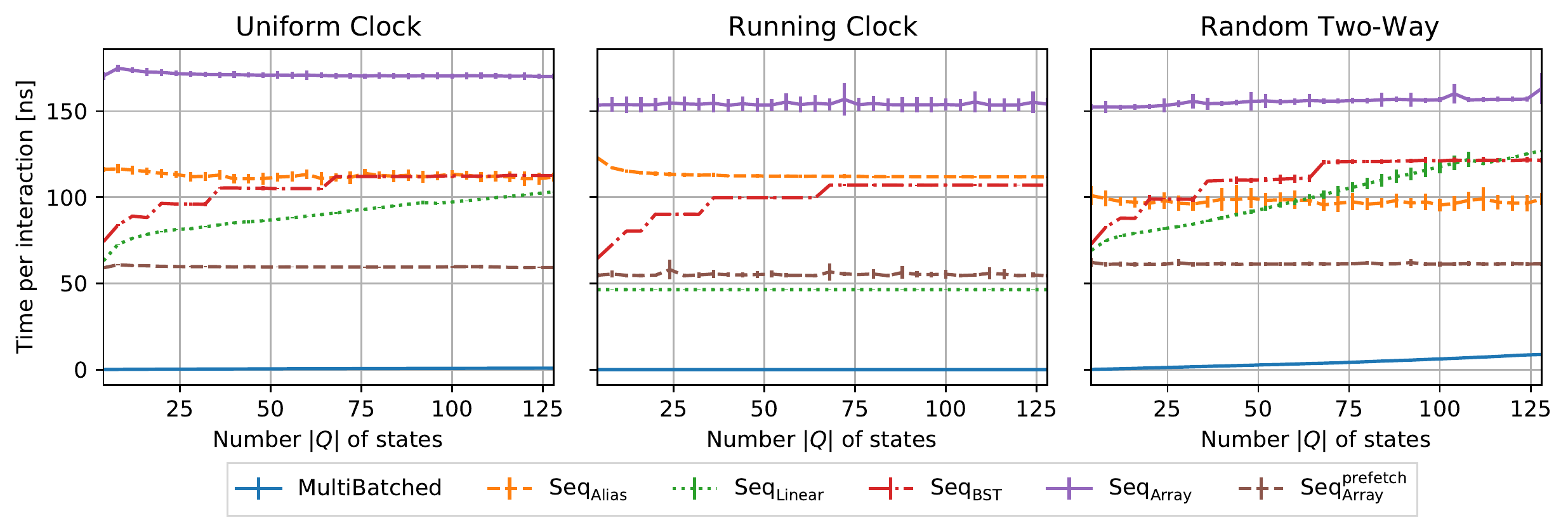}
    \caption{Processing time per interaction as function of the number of states $\q$ with $n=2^{30}$.}
    \label{fig:benchmark_states}
\end{figure}

As summarized in \cref{tab:simulators}, the number $\q$ of states crucially affects the algorithms' runtimes.
To quantify the practical impact, we carry out scaling experiments for $4 \le \q \le 128$ while fixing all remaining parameters.
\Cref{fig:benchmark_states} visualizes the results.

\algMulti performs best in all cases supporting our previous analysis.
It shows almost no scaling behavior for both clock protocols.
For \emph{Random Two-Way}, the algorithm is almost one order of magnitude faster than its competitors despite a slow-down of $40$ between the smallest and largest state sizes.

\algStepWiseArrayPre is the second fastest solution in almost all settings.
In the \emph{Running Clock} campaign, it is however outperformed by \algStepWiseLinear.
This can be explained by the fact that initially almost all agents are in state~$0$ which results in a constant time look-up despite the usage of a linear search.
This behavior motivates the renaming heuristic.

We observe no systematic dependency on $\q$ for \algStepWiseAlias rendering it a good choice for very large state spaces.
While it is up to a factor of $2.0$ slower compared to \algStepWiseArrayPre, the algorithm might be preferable in a parallel setting.

\paragraph{Memory Footprint and Parallelism}
Due to the stochastic nature of the protocols, we expect that in almost all applications several runs of the same protocol are required to derive statistically significant results.
On modern machines with many processor cores, one should be able to maximize the throughput by executing multiple independent simulations in parallel.
As visualized in \cref{fig:benchmark_selection_threading}, most simulators scale well with the number of threads and typically achieve a self-speedup of $40$ to $50$ times using 40 cores (plus HyperThreading) at $n=2^{30}$.
A notable exception is \algStepWiseArrayPre, which reaches only a speedup of $30$ as it saturates the memory controllers of both CPU sockets.

Another aspect of parallel execution is the memory footprint.
Since \algStepWiseArray requires constant memory per agent, our implementations of \algStepWiseArray and \algStepWiseArrayPre allocate in excess of \SI{320}{\giga\byte} main memory to execute 80~processes in parallel.
Although, this number can be reduced by constant factors using a more efficient representation of states, it has to be contrasted to the competing algorithms with a state space of only a few kilobytes\footnote{
    We report no exact numbers as the system's process overheads exceed the simulators' internal states.
} for the same campaign.

\begin{figure}
    \includegraphics[width=\textwidth]{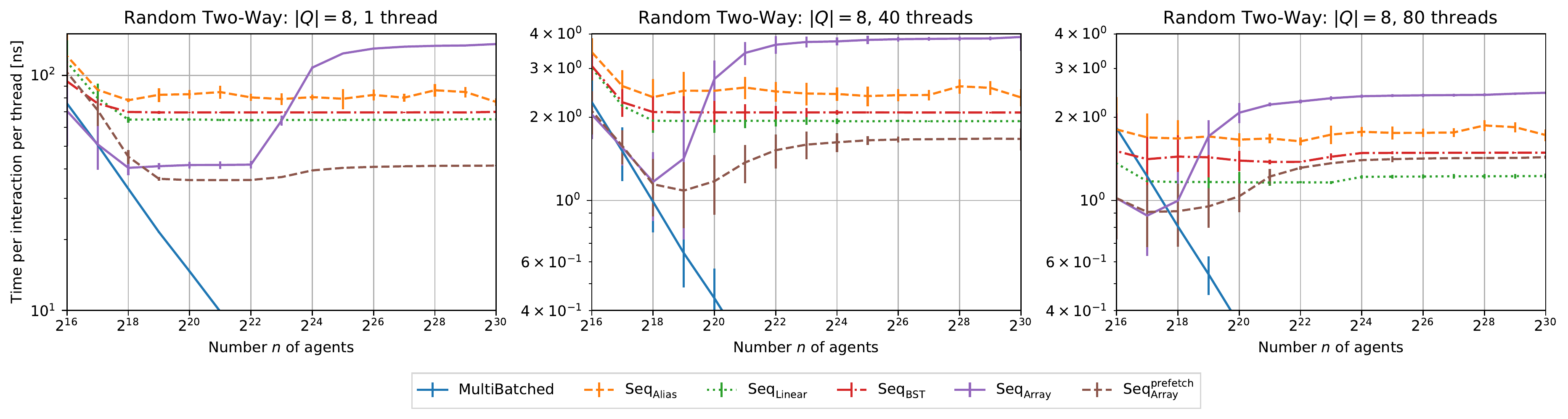}
    \caption{
        Effect of process parallelism on a machine with 40 CPU cores (plus HyperThreading).
    }
    \label{fig:benchmark_selection_threading}
\end{figure}

\section{Conclusions and Open Problems}
\label{sec:open-problems}

We considered the simulation of large population protocols to allow the experimental investigation of slowly scaling observables in such systems.
Two algorithm classes are discussed.

\emph{Sequential simulators} carry out each interaction one after the other.
We analyze the variants \algStepWiseArray, \algStepWiseLinear, \algStepWiseBST, and \algStepWiseAlias which differ in the data structures maintaining the agents' states, and demonstrate substantial differences in their practical performances.
As a by-product, we describe the \dat which might be independently applicable.

\emph{Batched simulators} coalesce interactions to achieve asymptotical speed-ups for protocols with a limited number of states.
Our implementation then simulates more than $2^{50}$ interactions in \SI{400}{\second} which is several orders of magnitudes larger than the fastest sequential simulator.

\paragraph{Possible Extensions}
In some variants of the population model it is assumed that interactions are limited to some communication network.
Since we store the configurations in our batched simulators as a multiset, it is not clear how to directly adapt our approach.
We believe this might be an interesting extension of our work.

Further variants are concerned with the way how agents interact.
It is straightforward to adapt our simulator software to a setting where, \eg a random matching of agents interacts in each time step.
Furthermore, our approach could be generalized to a setting where more than two agents interact.
In this case we are in need of good heuristics for partitioning the transition-tensor, since the work for updating a batch grows exponentially in the number of interacting agents.
In general, sampling the interaction counters is a frequent and costly task.
Therefore, any improved heuristic to sample from the $\q \times \q$ matrix using only $\oh(\q^2)$ variates would yield a measureable benefit for the total runtime of a simulation.

   \inshort{\clearpage}

   \printbibliography
   \clearpage

   \appendix
   \section*{Appendix}

\ifshort
\section{\dat{s}}
\label{apx:dynamic-alias-table}
In this section we give the full proof of the correctness of our \dat data structure.
Formally, we show \cref{thm:dynamic-alias-table}, which is restated for convenience as follows.
\restateThmDynamicAliasTable*

\def\walker{DBLP:journals/toms/Walker77}
\def\vose{DBLP:journals/tse/Vose91}

\medskip

\noindent Before we prove the theorem, we briefly recall the main ideas used for our data structure.

\medskip

Our goal is to model an urn which initially contains $n$ marbles, each of which has one of $k$ possible colors.
We assume that the colors are identified by numbers in $\set{1, \dots, k}$.
The urn defines a probability distribution $D$ for the color of a marble drawn uniformly at random: let $p_i$ be the probability that we sample a marble of color $i$.
We can use the alias method \cite{\walker} to sample from $D$ in constant time as follows.

\paragraph{Original Alias Method}
The alias method \cite{\walker} uses a table with two columns and $k$ rows, one row for each element (color) in the distribution $D$.
Each row $i$ has two entries corresponding to two elements.
Each element has a weight in $[0,1]$, and the two weights sum up to $1$ in each row.
The first element of row $i$ is always element $i$.
It has weight $F[i]$.
The second element of row $i$ is stored in $A[i]$.
It has a weight of $1 - F[i]$.
This means that the original alias method uses only the two arrays, $F$ and $A$, to store the distribution $p_1,\dots,p_k$.

To sample from $D$ in the original alias method, we first sample a row $i$ uniformly at random from $\set{1, \dots, k}$.
Then we draw a random real $X \in [0, 1)$.
If $X < F[i]$, we return the left element, $i$.
Otherwise, we return the right element, $A[i]$.
In the following, we modify the alias method and call the resulting data structure \dat.

\paragraph{\dat{s}}
Recall that we assume that our distribution corresponds to an urn storing $n$ marbles of $k$ different colors.
First, we explicitly add a second weight array $S[i]$ which stores the weight of the second column.
Now instead of storing just one real value $F[i]$ for each row $i$, we store the exact numbers of marbles as integers for the first and the second entry of row $i$ in $F[i]$ and $S[i]$, respectively.
As before, the first entry of row $i$ corresponds to color $i$ and the second entry of row $i$ corresponds to color $A[i]$.
The rows are constructed in such a way that the total weight of each row no longer adds up to the real value $1$, but to the integer value $\floor{n/k}$ or $\ceil{n/k}$ such that all rows in total add up to $n$.

\begin{observation} \label{obs:rebuild-dynamic-alias-table}
Let $U$ be a \dat encoding an urn with $n$ marbles and $k$ colors.
The data structure $U$ can be constructed in $\Oh(k)$ time.
\end{observation}

\begin{proof}
The algorithm by \textcite{\vose} can generate a (original) alias table representation of such a discrete probability distribution $D$ in $\Oh(k)$ time.
It is straightforward to define a mapping between the weights of the original alias method as computed in \cite{\vose} and the two integer values used in our \dat.
It follows that the \dat (using integer weights) can be constructed in $\Oh(k)$ time.
\end{proof}

\paragraph{Updating and Sampling from \dat{s}}
As already observed in \cref{sec:sequential-simulation}, our modified data structure now allows us to sample elements with and without replacement.
As before, let $R[i] = F[i] + S[i]$ denote the weight of row~$i$ and define $R_\text{min}$ and $R_\text{max}$ as smallest and largest row weights, respectively.
Observe that in general $R_\text{min} \neq R_\text{max}$.

In order to sample from $U$, we first select row~$i$ uniformly at random from $\set{1,\dots,k}$.
Then we draw a uniform variate~$X$ from $\set{0, \dots, R_\text{max}-1}$.
There are three possible events:
If $X < F[i]$, we emit the first element $i$.
If $F[i] \leq X < R[i]$, we emit the second element $A[i]$.
Otherwise, we reject the trial and restart the sampling process.

If we sample from $U$ without replacement, we decrement the weight of the element we just sampled.
This is always possible, since only elements with strictly positive weights can be sampled in the first place.
If we add a new element with color $i$ to $U$, we increment the weight of the first element of row $i$.

In order to guarantee expected constant sampling time, we ensure that the fraction
between $R_\text{min}$ and $R_\text{max}$ does not exceed a certain value.
Let $0 < \alpha < 1$ and $\beta > 1$ be two parameters chosen such that $\beta / \alpha = \Oh(1)$.
After each update to $U$ we require \begin{equation} \label{eq:dynamic-alias-table-condition} \alpha \floor{n/k} \le R_\text{min} \le R_{max} \le \beta \ceil{n/k} . \end{equation}
Otherwise, we rebuild the data structure in $\Oh(k)$ time.

\medskip

\noindent We are now ready to show \cref{thm:dynamic-alias-table}.

\begin{proof}[Proof of \cref{thm:dynamic-alias-table}]
We start with the memory complexity.
The \dat $U$ stores the values of $k$, $n$, and $R_\text{max}$ as well as three arrays.
Array $F[1\dots k]$ stores the weight of the first column, array $S[1\dots k]$ stores the weight of the second column, and array $A[1\dots k]$ stores the alias, \ie the element of the second column.
All entries are integers from $\set{0, \ldots, n}$ (recall that we assume $k \leq \sqrt{n}$).
Thus, the \dat requires $\Theta(k \log n)$ bits of memory.

\medskip

Let us now consider the sampling procedure.
First, we consider the rejection probability.
Recall that we first sample a row $i$ and then draw a uniform variate $X$ from $\set{0, \dots, R_\text{max}-1}$.
As before, we denote the total weight of row $i$ as $R[i]$ with $R[i] = F[i] + S[i]$.
A sampling trial in row $i$ is rejected if $X \geq R[i]$, \ie with probability $R[i] / R_\text{max}$.
Therefore, the probability to reject a sample from any row is at most $R_\text{min} / R_\text{max}$.
From the conditions in \cref{eq:dynamic-alias-table-condition} we get that the rejection probability is at most $\alpha / \beta = \Oh(1)$ and, conversely, we have at least a constant success probability of $(\beta - \alpha) / \beta$.
The number of trials until we emit an element is therefore geometrically distributed and has an expected value of at most $\beta / (\beta - \alpha) = \Oh(1)$.

\def\R{\mathcal{R}}
\def\S{\mathcal{S}}
\def\OMEGA{\mathbf{\Omega}}
It remains to show that we emit an element of color $i$ with probability $p_i = \hat C_i / n$, where $\hat C_i$ is the number of marbles of color $i$ in the \dat $U$.
We consider a single sampling trial.
Observe that in each trial we are given a uniform probability space
$\OMEGA = \set{(i, x)\ :\ 1 \leq i \leq k \text{ and } 0 \leq x < R_\text{max}}$.
From this probability space we draw the row $i$ and the value $X$ uniformly at random.
Fix a color $c$ and let $\S_c$ be the set of all events $(i,x)$ which lead to emission of an element of color $c$ for this probability space $\OMEGA$.
An event $(i,x)$ is in $\S_c$ if and only if (i)
$i = c$ and $x < F[i]$ or (ii) $A[i] = c$ and $F[i] \leq x < F[i] + S[i]$.

The \dat $U$ is constructed such that the total weight for each color $c$
always equals $\hat C_c$.
Therefore, counting all elementary events gives us $|\S_c| = \hat C_c$.
Observe that $\OMEGA$ is a uniform probability space since the row $i$ and the value $X$ are drawn uniformly. It has size $|\OMEGA| = k R_\text{max}$.
Hence, all events in $\S_c$ have equal probability $1/(k R_\text{max})$, and we get $\Prob{\S_c} = |\S_c| / (k R_\text{max}) = \hat C_c / (k R_\text{max})$.

Let $\mathcal{\R}$ be the event that a trial is rejected.
Analogously to before, we enumerate over all elementary events and obtain $|\R| =
\sum_{i} (R_\text{max} - R[i]) = k R_\text{max} - n$.
For the complementary event $\overline\R$ we get $|\overline\R| = n$, which matches the intuition that the urn contains $n$ marbles.
Hence, we have $\Prob{\overline\R} = n / (k R_\text{max})$.
Observe that $\S_c$ and $\R$ are mutually exclusive and hence $\S_c \cap \overline\R = \S_c \setminus \R = \S_c$.

Rejected trials emit no element, but are repeated.
Hence, we condition on $\overline\R$ and obtain 
\begin{align*}
p_c &= \Prob{\S_c | \overline\R}
= \frac{\Prob{\S_c \cap \overline\R} }{ \Prob{\overline\R} }
= \frac{\Prob{S_c}}{\Prob{\overline\R}}
= \frac{\hat C_c}{k R_\text{max}} \cdot \frac{k R_\text{max}}{n}
= \frac{\hat C_c}{n} .
\end{align*}
This means that a color $c$ is indeed emitted with the correct probability $p_c = \hat C_c / n$.

\medskip

Finally, we consider the amortized costs of rebuilding the \dat $U$ ever so often.
Observe that $U$ has to be rebuilt whenever the condition in \cref{eq:dynamic-alias-table-condition} is violated.
This can happen in two possible ways.
\begin{itemize}
\item Case 1: $R_\text{min} < \alpha \floor{n/k}$.

\noindent
In this case there must exist a row $i$ for which $R[i] < \alpha \floor{n/k}$.
Observe that by assumption of the theorem we have $n \geq k^2$, and after rebuilding $U$ we have $R_\text{min} = \floor{n/k}$.
In order to have $R[i] < \alpha \floor{n/k}$, at least $\floor{n/k}- \alpha \floor{n/k} \geq k(1-\alpha)$ elements must have been deleted from row $i$.
As $0 < \alpha < 1$, this happens only after at least $k(1-\alpha) = \Omega(k)$ sampling operations.
Now according to \cref{obs:rebuild-dynamic-alias-table}, rebuilding takes time $\Oh(k)$.
Together this implies that rebuilding $U$ takes amortized constant time per update of $U$.

\item Case 2: $R_\text{max} > \beta \ceil{n/k}$.

\noindent The second case follows analogously to the first case.
If $R_{\max} > \beta \ceil{n/k}$, at least $\beta \ceil{n/k} - \ceil{n/k} \geq k(\beta-1)$ elements must have been added.
This takes at least $k(\beta - 1) = \Omega(k)$ insertions, and hence rebuilding $U$ takes amortized constant time per insertion.
\end{itemize}

\noindent Removal and insertion operations can be arbitrarily mixed and interact only beneficially towards the amortization arguments.
This concludes the proof of the theorem.
\end{proof}

 \clearpage

{
\section{Batch Processing}
\label{apx:batch-processing}
In this appendix, we give the pseudocode for our simulator \algSimple in \cref{alg:seq-simple} and an illustration of collision-free runs in \cref{fig:batched-interactions}.

\raggedbottom
\begin{figure}[H]
    \centering
    \scalebox{1}{%
\begin{tikzpicture}
   \def\is{*4.2em}

   \node[anchor=west] at (0em, 5em) {\textbf{The original interaction sequence (\cf \cref{{sec:sequential-simulation}}):}};
   \foreach \a/\b/\k/\l [count=\i] in {3/2/1/2, 1/2/3/4, 1/1/5/6, 2/1/7/8, 3/3/9/10, 1/2/11/12, d/d/d/d, 1/1/{2\ell-1}/{2\ell}, 2/?/{2\ell+1}/6} {
      \ifthenelse{\equal{\a}{d}}{
         \node at (\i \is, 0) {$\dots$};
      }{
         \node[draw, circle, inner sep=0, minimum height=1.5em, state\a] (rqa\i) at (\i \is - 0.22 \is, 0) {$q_\a$};
         \node[draw, circle, inner sep=0, minimum height=1.5em, state\b] (rqb\i) at (\i \is + 0.22 \is, 0) {$q_\b$};
         \node[anchor=south] at (rqa\i.north) {$s_{i_{\k}}$};
         \node[anchor=south, onif={<\equal{\l}{6}> conflict}] at (rqb\i.north) {$s_{i_{\l}}$};
      }
   }

   \node[draw, minimum width=0.5\is, minimum height=0.5\is, anchor=north] (trans) at (3\is, -0.5\is) {$\delta$};
   \path[draw, semithick, green, arrows={-\arrowHead}] (rqa3) to (trans);
   \path[draw, semithick, green, arrows={-\arrowHead}] (rqb3) to (trans);

   \path[draw, red, semithick, arrows={-\arrowHead}] (trans.east) -| node[pos=0.25, black, font=\small, align=center] {
      the updated state of the agent drawn twice is\\
      known only after the $\delta$ was evaluated} (rqb9);
   \path[draw, red, semithick, arrows={-\arrowHead}, bend right] (trans.east) to (rqb3);
   \path[draw, red, semithick, arrows={-\arrowHead}, bend left] (trans.west) to (rqa3);

   \path[draw] ($(rqa1.west) + (-0.3em, 0)$) to ++(0, 2.5em) -| node [pos=0.25, above, yshift=-0.3em] {
      the $\ell$ independent interactions can be rearranged arbitrarily} ($(rqb8.east) + (0.3em, 0)$);

\node[anchor=west] at (0em, -6em) {\textbf{After sorting state pairs:}};
   \node[anchor=north, minimum width=2\is, minimum height=16em, fill=black!5, align=center] at (8.5\is, -6em) {\ \\[0.5em] \small special treatment\\[-0.4em] \small for collision};

   \foreach \a/\b/\k/\l [count=\i] in {1/1/20/21, 1/1/{2\ell-1}/{2\ell}, 1/2/3/4, 1/2/11/12, 2/1/7/8, d/d, 3/3/9/10, 1/1/5/6, 2/?/2\ell+1/6} {
      \ifthenelse{\equal{\a}{d}}{
         \node at (\i \is, -9em) {$\dots$};
      }{
         \node[draw, circle, inner sep=0, minimum height=1.5em, state\a] (rrqa\i) at (\i \is - 0.22 \is, -9em) {$q_\a$};
         \node[draw, circle, inner sep=0, minimum height=1.5em, state\b] (rrqb\i) at (\i \is + 0.22 \is, -9em) {$q_\b$};
         \node[anchor=south] at (rrqa\i.north) {$s_{i_{\k}}$};
         \node[anchor=south, onif={<\equal{\l}{6}> conflict}] at (rrqb\i.north) {$s_{i_{\l}}$};
      }
   }

   \node[draw, minimum width=0.5\is, minimum height=0.5\is, anchor=north] (trans1) at (8\is, -0.5\is - 9em) {$\delta$};
   \path[draw, red,   semithick, arrows={-\arrowHead}] (trans1.east) -| (rrqb9);
   \path[draw, green, semithick, arrows={-\arrowHead}] (rrqa8) to (trans1);
   \path[draw, green, semithick, arrows={-\arrowHead}] (rrqb8) to (trans1);
   \path[draw, red,   semithick, arrows={-\arrowHead}, bend right] (trans1.east) to (rrqb8);
   \path[draw, red,   semithick, arrows={-\arrowHead}, bend left]  (trans1.west) to (rrqa8);

   \node[anchor=west] at (0em, -15em) {\textbf{After merging interactions with identical state pairs:}};

   \foreach \a/\b [count=\i] in {-/1, -/1, -/2, -/2, -/3, -/3, -/d, 1/1, 2/?} {
      \ifthenelse{\equal{\a}{d}}{
         \node at (\i \is, -17em) {$\dots$};
      }{\ifthenelse{\equal{\a}{-}}{}{
            \node[draw, circle, inner sep=0, minimum height=1.5em, state\a] (gqa\i) at (\i \is - 0.22 \is, -17em) {$q_\a$};
            \node[draw, circle, inner sep=0, minimum height=1.5em, state\b] (gqb\i) at (\i \is + 0.22 \is, -17em) {$q_\b$};
      }}
   }

   \foreach \a/\y in {1/0, 2/2,3/4} {
      \foreach \b [count=\i] in {1, 2, 3} {
         \node[draw, circle, inner sep=0, minimum height=1.5em, state\a, opacity=0.5] (gqa\i\a) at (2* \i \is - 0.22 \is, -17em - \y em) {$q_\a$};
         \node[draw, circle, inner sep=0, minimum height=1.5em, state\b, opacity=0.5] at (2* \i \is + 0.22 \is, -17em - \y em) {$q_\b$};
         \node[anchor=east] at (gqa\i\a.west) {$\mathbf{d_{\a\b}}$ \textcolor{black!50}{$\times$}};
   }}

   \node at (1.5em, -19em) {$D = $ \huge $\Bigg($};
   \node at (28em, -19em) {\huge $\Bigg)$};

   \node[draw, minimum width=0.5\is, minimum height=0.5\is, anchor=north] (trans1) at (8\is, -0.5\is - 17em) {$\delta$};
   \path[draw, red,   semithick, arrows={-\arrowHead}] (trans1.east) -| (gqb9);
   \path[draw, green, semithick, arrows={-\arrowHead}] (gqa8) to (trans1);
   \path[draw, green, semithick, arrows={-\arrowHead}] (gqb8) to (trans1);
   \path[draw, red,   semithick, arrows={-\arrowHead}, bend right] (trans1.east) to (gqb8);
   \path[draw, red,   semithick, arrows={-\arrowHead}, bend left]  (trans1.west) to (gqa8);
\end{tikzpicture} %
}
    \caption{%
        Batch processing uses collision-free runs, long sequences of independent interactions, which can be rearranged and grouped together.%
	}
    \label{fig:batched-interactions}
\end{figure}

\ifshort
\begin{algorithm}[H]
\else
\begin{algorithm}[t]
\fi
	\caption{\algSimple: The algorithmic framework for simulation in batches.}
	\label{alg:seq-simple}

	\SetKwComment{Comment}{$\triangleright$\ }{}
	\DontPrintSemicolon
	\SetKwInOut{Input}{input\expandafter\?\?}

	\footnotesize

	\Input{~%
		configuration $C$,
		transition function $\delta$,
		number of steps $N$
	}

	$t \gets 0$\;
	\While{t < N}{

	$\ell \gets$ sample length of a collision-free run\;
	\smallskip
	let $D = (d_{ij})$ be a $\q \times \q$ matrix and sample $d_{ij}$ as \Comment*[r]{batch processing} the number of interactions~$(q_i,q_j)$ among $\ell$ interactions\;
	\smallskip
	let $C'$ be an empty configuration\;
	\ForEach{$(q_i, q_j) \in Q^2$}{
		remove from $C$: $d_{ij}$ agents in states $q_i$, and $d_{ij}$ agents in states $q_j$\;
		$(q'_i, q'_j) \gets \delta(q_i, q_j)$\;
		add to $C'$: $d_{ij}$ agents in states $q'_i$, and  $d_{ij}$ agents in states $q'_j$\;
	}

		\eIf(\Comment*[f]{plant a collision}){$\ell$ is even}{
			sample agent $c_1$ without replacement from $C'$\Comment*[r]{collision at $c_1$}
			merge $C'$ into $C$\;
			sample agent $c_2$ without replacement from $C$\;
		}
		{
			sample agent $c_1$ without replacement from $C$\;
			sample agent $c_2$ without replacement from $C'$\Comment*[r]{collision at $c_2$}
			merge $C'$ into $C$\;
		}
		add agents $\delta(c_1,c_2)$ to $C$\;
	$t \gets t + \ell + 1$\;
	}

	\normalsize
\end{algorithm}
 }

\section{Omitted Proofs}
\ifshort
  \subsection{Correctness of Reordering Sampling}
\else
  \section{Correctness of Reordering Sampling}
\fi
\label{app:reordering-argument}
In this section, we show that the batching steps of \algSimple and \algMulti preserve the population model's interaction probabilities.
The proof relies on the fact that in the underlying sampling process yields all permutations of colors with equal probability.
This follows from the following observation regarding a special case of  P\'{o}lya urn style sampling~\cite{AIHP_1930__1_2_117_0}.
We are given an urn that contains $r$ red marbles and $g$ green marbles.
From this urn, we draw a sequence of marbles independently and uniformly at random without replacement.
A well-known property is that the probability that the $i$-th marble drawn has a fixed color depends only on the initial numbers of marbles in the urn.
For completeness, we include a brief proof of this property for an urn with two colors.
\begin{lemma}\label{lem:polya-urn-reorder}
  Consider an urn with initially $r$ red and $g$ green marbles, and let $R_i$ ($G_i$) be the event that the $i$-th marble drawn without replacement is red (green).
  Then,
  \begin{equation*}
    \prob{R_i} = \frac{r}{r+g} \quad \text{and} \quad \prob{G_i} = \frac{g}{r+g}.
  \end{equation*}
\end{lemma}
\begin{proof}
  This is shown by induction.
  Clearly, $\prob{R_1} = \frac{r}{r+g}$.
  Consider now $\prob{R_{i+1}}$.
  The preceding, $i$-th, draw can have yielded either red or green, so we can consider the two distinct cases conditioned on the preceding draw as:
  \begin{equation*}
    \prob{R_{i+1}} =
    \prob{R_{i+1} | R_{i}}\cdot\prob{R_{i}}\ + \
    \prob{R_{i+1} | G_{i}}\cdot\prob{G_{i}}.
  \end{equation*}
  Independent of the preceding draws' outcomes, there will be $r+g-i$ marbles remaining after $i$ draws.
  In the case of $R_{i}$, there will be one less red marble, and one more green marble compared to $G_{i}$.
  Let $r_{i}$ ($g_{i}$) be the number of remaining red (green) marbles \emph{before} the preceding $i$-th marble was drawn.
  Depending on the color of the $i$-th marble, the number of red or green marbles \emph{after} the $i$-th draw differs by one and therefore,
  \begin{equation*}
    \prob{R_{i+1} | R_{i}} =
    \frac{r_{i}-1}{r_{i}+g_{i}-1}
    \quad \text{and} \quad
    \prob{R_{i+1} | G_{i}} =
    \frac{r_{i}}{r_{i}+g_{i}-1}.
  \end{equation*}
  Combining this, we get:
  \begin{equation*}
    \prob{R_{i+1}} =
    \frac{r_{i}-1}{r_{i}+g_{i}-1}\cdot\frac{r_{i}}{r_{i}+g_{i}} +
    \frac{r_{i}}{r_{i}+g_{i}-1}\cdot\frac{g_{i}}{r_{i}+g_{i}} =
    \frac{r_{i}(r_{i}-1+g_{i})}{(r_{i}+g_{i}-1)(r_{i}+g_{i})} =
    \frac{r_{i}}{r_{i} + g_{i}},
  \end{equation*}
  which is identical to $\prob{R_i}$ and thus by induction, $\prob{R_{i+1}}=\prob{R_{i}}= r / (r+g)$.
\end{proof}

A generalization of \cref{lem:polya-urn-reorder} to more colors is straightforward by fixing a color~$c$ of interest.
Then, color~$c$ takes the role of red, while all other colors are coalesced into green.

\begin{lemma}
  \label{lem:batch-sampling-correctness}
  \algSimple's and \algMulti's sampling algorithm for matrix~$D$ correctly simulates drawing and counting of $\ellHalf$ state pairs without replacement.
\end{lemma}

\begin{proof}
  When sampling from the matrix $D$ in \cref{sec:batch-processing}, we exploit this independence on previous draws when essentially reordering our sampling.
  We do so using two observations:
  \begin{itemize}
  \item Sampling first the initiators and then the responders is just reordering the draws.
  \item When sampling either of the two groups, we simply sample frequencies of the different colors.
    Due to \cref{lem:polya-urn-reorder} all permutations with the same frequency counts are equally likely to occur (\ie the probability of any sequence only depends upon the frequency counts of the colors sampled).
    We hence do not alter the distribution by sampling frequencies one color at a time. \qedhere
  \end{itemize}
\end{proof}
 \ifshort \clearpage \fi
\subsection{Length of a Collision-Free Run}
\label{sec:length-of-a-collision-free-run}
\restateLemCollisionDistance*
\proofCollisionDistance

\restateLemExpectedBatchLength*
\proofExpectedBatchLength

\restateExpectedBatchLengthPrescribed*
\proofExpectedBatchLengthPrescribed

\restateBatchLengthProb*
\proofBatchLengthProb

\restateMultiBatchNumberInteractions*
\proofMultiBatchNumberInteractions

\restateAlgMultiRunningTime*
\proofAlgMultiRunningTime

\section{Implementation Details}
\label{apx:implementation-details}
\subsection{Simulator Interface}
\begingroup
\providecommand\nolinenumbers{}
\nolinenumbers
\begin{lstlisting}[
	caption={Example of a simple \problem{Leader Election} protocol.},
	label=lst:leader-election,
]
struct LeaderElectionProtocol : public
  Protocols::DeterministicProtocol, Protocols::OneWayProtocol
{
  enum Roles : state_t {Follower = 0, Leader = 1};

  state_t operator() (state_t initiator, const state_t responder) const {
    if (initiator == responder) initiator = Follower;
    return initiator; // simplification as one-way protocol
  }

  state_t num_states() const {return 2;}
};
\end{lstlisting}
\begin{lstlisting}[
	caption={%
		A non-deterministic protocol in which either the initiator or the responder circularly increments its state by one. In each interaction the active one is determined by tossing a fair coin.
	},
	label=lst:non-deterministic,
]
state_t operator() (state_t initiator, state_t responder,
                    count_t num, auto& assign) const {
  count_t n = std::binomial_distribution<>{num, 0.5}(gen_);

  //    (state                       , number of agents);
  assign(initiator + 1 %
  assign(responder                   , n);

  assign(initiator                   , num - n); // responder advances
  assign(responder + 1 %
}
\end{lstlisting}
\endgroup
 
\subsection{Sampling the Length of a Collision-Free Run}
\label{sec:sampling-length-collision-free-run}
Recall that \algSimple and \algMulti repeatedly sample the length~$\ell$ of a collision-free run.
In the following we discuss how this sampling can be implemented using the inverse sampling technique (\IST)~\cite{DBLP:books/sp/Devroye86}:
let $\cdf(x)$ be the cumulative density function of a target distribution.
Then, \IST draws a uniform variate~$U$ from $[0; 1]$, solves $U = \cdf(x)$ for $x$, and returns it as the sample.
We denote the CDF of $\colDistr{n}{k}$ as $\colCDF{n}{k}{t}$.
\cref{lem:collision-distance} yields:
\begin{equation*}
	1 - \colCDF{n}{k}{t}
	  = \prob{\ell {>} t}
	  = \prod_i^{t} \frac{n {-} k {-} i}{n}
	  = \frac{1}{n^t} \frac{ (n-k)! }{ (n-k-t-1)! }
	  \stackrel{(x-1)! = \Gamma(x)}{=} n^{-t} \frac{\Gamma(n - k + 1)}{\Gamma(n-k-t)}.
\end{equation*}
Since we are not aware of an inverse that can be evaluated fast, we numerically solve $U = \colCDF{n}{k}{t}$ for $t$.
To avoid numerical instabilities, we rewrite the expression in terms of $\log\Gamma(x)$, which is available as the C standard function \texttt{lgamma}$(x)$:
\begin{equation*}
	U = 1 - n^{-t} \frac{\Gamma(n - k + 1)}{\Gamma(n-k-t)}
	\quad \Leftrightarrow\quad
	\log\left(1 - U\right) = \log\Gamma\left(n - k + 1\right) - \log\Gamma\left(n-k-t\right) - t \log n
\end{equation*}

Lacking a cheap derivative of the RHS, we rely on first-order numerical inversion methods only.
In this context, an ad-hoc combination of binary search and regula-falsi gave most consistent results.
We jump-start the search using a small look-up table containing lower and upper bounds on~$t$ for intervals of $U$ and $k$.

While the method requires $\Oh(\log n)$ evaluations of $\colCDF{n}{k}{\cdot}$, we observe less than ten calls on average for $n = 2^{50}$.
The resulting sampling algorithm has a practical runtime comparable to the sampling of hypergeometric random variates.
Since the latter is sampled much more frequently, further optimizations will yield limited results to the total runtimes of $\algSimple$ and $\algMulti$.
 
\else

\section{\dat{s}}
\label{apx:dynamic-alias-table}
In this section we give the full proof of the correctness of our \dat data structure.
Formally, we show \cref{thm:dynamic-alias-table}, which is restated for convenience as follows.
\restateThmDynamicAliasTable*

\def\walker{DBLP:journals/toms/Walker77}
\def\vose{DBLP:journals/tse/Vose91}

\medskip

\noindent Before we prove the theorem, we briefly recall the main ideas used for our data structure.

\medskip

Our goal is to model an urn which initially contains $n$ marbles, each of which has one of $k$ possible colors.
We assume that the colors are identified by numbers in $\set{1, \dots, k}$.
The urn defines a probability distribution $D$ for the color of a marble drawn uniformly at random: let $p_i$ be the probability that we sample a marble of color $i$.
We can use the alias method \cite{\walker} to sample from $D$ in constant time as follows.

\paragraph{Original Alias Method}
The alias method \cite{\walker} uses a table with two columns and $k$ rows, one row for each element (color) in the distribution $D$.
Each row $i$ has two entries corresponding to two elements.
Each element has a weight in $[0,1]$, and the two weights sum up to $1$ in each row.
The first element of row $i$ is always element $i$.
It has weight $F[i]$.
The second element of row $i$ is stored in $A[i]$.
It has a weight of $1 - F[i]$.
This means that the original alias method uses only the two arrays, $F$ and $A$, to store the distribution $p_1,\dots,p_k$.

To sample from $D$ in the original alias method, we first sample a row $i$ uniformly at random from $\set{1, \dots, k}$.
Then we draw a random real $X \in [0, 1)$.
If $X < F[i]$, we return the left element, $i$.
Otherwise, we return the right element, $A[i]$.
In the following, we modify the alias method and call the resulting data structure \dat.

\paragraph{\dat{s}}
Recall that we assume that our distribution corresponds to an urn storing $n$ marbles of $k$ different colors.
First, we explicitly add a second weight array $S[i]$ which stores the weight of the second column.
Now instead of storing just one real value $F[i]$ for each row $i$, we store the exact numbers of marbles as integers for the first and the second entry of row $i$ in $F[i]$ and $S[i]$, respectively.
As before, the first entry of row $i$ corresponds to color $i$ and the second entry of row $i$ corresponds to color $A[i]$.
The rows are constructed in such a way that the total weight of each row no longer adds up to the real value $1$, but to the integer value $\floor{n/k}$ or $\ceil{n/k}$ such that all rows in total add up to $n$.

\begin{observation} \label{obs:rebuild-dynamic-alias-table}
Let $U$ be a \dat encoding an urn with $n$ marbles and $k$ colors.
The data structure $U$ can be constructed in $\Oh(k)$ time.
\end{observation}

\begin{proof}
The algorithm by \textcite{\vose} can generate a (original) alias table representation of such a discrete probability distribution $D$ in $\Oh(k)$ time.
It is straightforward to define a mapping between the weights of the original alias method as computed in \cite{\vose} and the two integer values used in our \dat.
It follows that the \dat (using integer weights) can be constructed in $\Oh(k)$ time.
\end{proof}

\paragraph{Updating and Sampling from \dat{s}}
As already observed in \cref{sec:sequential-simulation}, our modified data structure now allows us to sample elements with and without replacement.
As before, let $R[i] = F[i] + S[i]$ denote the weight of row~$i$ and define $R_\text{min}$ and $R_\text{max}$ as smallest and largest row weights, respectively.
Observe that in general $R_\text{min} \neq R_\text{max}$.

In order to sample from $U$, we first select row~$i$ uniformly at random from $\set{1,\dots,k}$.
Then we draw a uniform variate~$X$ from $\set{0, \dots, R_\text{max}-1}$.
There are three possible events:
If $X < F[i]$, we emit the first element $i$.
If $F[i] \leq X < R[i]$, we emit the second element $A[i]$.
Otherwise, we reject the trial and restart the sampling process.

If we sample from $U$ without replacement, we decrement the weight of the element we just sampled.
This is always possible, since only elements with strictly positive weights can be sampled in the first place.
If we add a new element with color $i$ to $U$, we increment the weight of the first element of row $i$.

In order to guarantee expected constant sampling time, we ensure that the fraction
between $R_\text{min}$ and $R_\text{max}$ does not exceed a certain value.
Let $0 < \alpha < 1$ and $\beta > 1$ be two parameters chosen such that $\beta / \alpha = \Oh(1)$.
After each update to $U$ we require \begin{equation} \label{eq:dynamic-alias-table-condition} \alpha \floor{n/k} \le R_\text{min} \le R_{max} \le \beta \ceil{n/k} . \end{equation}
Otherwise, we rebuild the data structure in $\Oh(k)$ time.

\medskip

\noindent We are now ready to show \cref{thm:dynamic-alias-table}.

\begin{proof}[Proof of \cref{thm:dynamic-alias-table}]
We start with the memory complexity.
The \dat $U$ stores the values of $k$, $n$, and $R_\text{max}$ as well as three arrays.
Array $F[1\dots k]$ stores the weight of the first column, array $S[1\dots k]$ stores the weight of the second column, and array $A[1\dots k]$ stores the alias, \ie the element of the second column.
All entries are integers from $\set{0, \ldots, n}$ (recall that we assume $k \leq \sqrt{n}$).
Thus, the \dat requires $\Theta(k \log n)$ bits of memory.

\medskip

Let us now consider the sampling procedure.
First, we consider the rejection probability.
Recall that we first sample a row $i$ and then draw a uniform variate $X$ from $\set{0, \dots, R_\text{max}-1}$.
As before, we denote the total weight of row $i$ as $R[i]$ with $R[i] = F[i] + S[i]$.
A sampling trial in row $i$ is rejected if $X \geq R[i]$, \ie with probability $R[i] / R_\text{max}$.
Therefore, the probability to reject a sample from any row is at most $R_\text{min} / R_\text{max}$.
From the conditions in \cref{eq:dynamic-alias-table-condition} we get that the rejection probability is at most $\alpha / \beta = \Oh(1)$ and, conversely, we have at least a constant success probability of $(\beta - \alpha) / \beta$.
The number of trials until we emit an element is therefore geometrically distributed and has an expected value of at most $\beta / (\beta - \alpha) = \Oh(1)$.

\def\R{\mathcal{R}}
\def\S{\mathcal{S}}
\def\OMEGA{\mathbf{\Omega}}
It remains to show that we emit an element of color $i$ with probability $p_i = \hat C_i / n$, where $\hat C_i$ is the number of marbles of color $i$ in the \dat $U$.
We consider a single sampling trial.
Observe that in each trial we are given a uniform probability space
$\OMEGA = \set{(i, x)\ :\ 1 \leq i \leq k \text{ and } 0 \leq x < R_\text{max}}$.
From this probability space we draw the row $i$ and the value $X$ uniformly at random.
Fix a color $c$ and let $\S_c$ be the set of all events $(i,x)$ which lead to emission of an element of color $c$ for this probability space $\OMEGA$.
An event $(i,x)$ is in $\S_c$ if and only if (i)
$i = c$ and $x < F[i]$ or (ii) $A[i] = c$ and $F[i] \leq x < F[i] + S[i]$.

The \dat $U$ is constructed such that the total weight for each color $c$
always equals $\hat C_c$.
Therefore, counting all elementary events gives us $|\S_c| = \hat C_c$.
Observe that $\OMEGA$ is a uniform probability space since the row $i$ and the value $X$ are drawn uniformly. It has size $|\OMEGA| = k R_\text{max}$.
Hence, all events in $\S_c$ have equal probability $1/(k R_\text{max})$, and we get $\Prob{\S_c} = |\S_c| / (k R_\text{max}) = \hat C_c / (k R_\text{max})$.

Let $\mathcal{\R}$ be the event that a trial is rejected.
Analogously to before, we enumerate over all elementary events and obtain $|\R| =
\sum_{i} (R_\text{max} - R[i]) = k R_\text{max} - n$.
For the complementary event $\overline\R$ we get $|\overline\R| = n$, which matches the intuition that the urn contains $n$ marbles.
Hence, we have $\Prob{\overline\R} = n / (k R_\text{max})$.
Observe that $\S_c$ and $\R$ are mutually exclusive and hence $\S_c \cap \overline\R = \S_c \setminus \R = \S_c$.

Rejected trials emit no element, but are repeated.
Hence, we condition on $\overline\R$ and obtain 
\begin{align*}
p_c &= \Prob{\S_c | \overline\R}
= \frac{\Prob{\S_c \cap \overline\R} }{ \Prob{\overline\R} }
= \frac{\Prob{S_c}}{\Prob{\overline\R}}
= \frac{\hat C_c}{k R_\text{max}} \cdot \frac{k R_\text{max}}{n}
= \frac{\hat C_c}{n} .
\end{align*}
This means that a color $c$ is indeed emitted with the correct probability $p_c = \hat C_c / n$.

\medskip

Finally, we consider the amortized costs of rebuilding the \dat $U$ ever so often.
Observe that $U$ has to be rebuilt whenever the condition in \cref{eq:dynamic-alias-table-condition} is violated.
This can happen in two possible ways.
\begin{itemize}
\item Case 1: $R_\text{min} < \alpha \floor{n/k}$.

\noindent
In this case there must exist a row $i$ for which $R[i] < \alpha \floor{n/k}$.
Observe that by assumption of the theorem we have $n \geq k^2$, and after rebuilding $U$ we have $R_\text{min} = \floor{n/k}$.
In order to have $R[i] < \alpha \floor{n/k}$, at least $\floor{n/k}- \alpha \floor{n/k} \geq k(1-\alpha)$ elements must have been deleted from row $i$.
As $0 < \alpha < 1$, this happens only after at least $k(1-\alpha) = \Omega(k)$ sampling operations.
Now according to \cref{obs:rebuild-dynamic-alias-table}, rebuilding takes time $\Oh(k)$.
Together this implies that rebuilding $U$ takes amortized constant time per update of $U$.

\item Case 2: $R_\text{max} > \beta \ceil{n/k}$.

\noindent The second case follows analogously to the first case.
If $R_{\max} > \beta \ceil{n/k}$, at least $\beta \ceil{n/k} - \ceil{n/k} \geq k(\beta-1)$ elements must have been added.
This takes at least $k(\beta - 1) = \Omega(k)$ insertions, and hence rebuilding $U$ takes amortized constant time per insertion.
\end{itemize}

\noindent Removal and insertion operations can be arbitrarily mixed and interact only beneficially towards the amortization arguments.
This concludes the proof of the theorem.
\end{proof}

\ifshort
  \subsection{Correctness of Reordering Sampling}
\else
  \section{Correctness of Reordering Sampling}
\fi
\label{app:reordering-argument}
In this section, we show that the batching steps of \algSimple and \algMulti preserve the population model's interaction probabilities.
The proof relies on the fact that in the underlying sampling process yields all permutations of colors with equal probability.
This follows from the following observation regarding a special case of  P\'{o}lya urn style sampling~\cite{AIHP_1930__1_2_117_0}.
We are given an urn that contains $r$ red marbles and $g$ green marbles.
From this urn, we draw a sequence of marbles independently and uniformly at random without replacement.
A well-known property is that the probability that the $i$-th marble drawn has a fixed color depends only on the initial numbers of marbles in the urn.
For completeness, we include a brief proof of this property for an urn with two colors.
\begin{lemma}\label{lem:polya-urn-reorder}
  Consider an urn with initially $r$ red and $g$ green marbles, and let $R_i$ ($G_i$) be the event that the $i$-th marble drawn without replacement is red (green).
  Then,
  \begin{equation*}
    \prob{R_i} = \frac{r}{r+g} \quad \text{and} \quad \prob{G_i} = \frac{g}{r+g}.
  \end{equation*}
\end{lemma}
\begin{proof}
  This is shown by induction.
  Clearly, $\prob{R_1} = \frac{r}{r+g}$.
  Consider now $\prob{R_{i+1}}$.
  The preceding, $i$-th, draw can have yielded either red or green, so we can consider the two distinct cases conditioned on the preceding draw as:
  \begin{equation*}
    \prob{R_{i+1}} =
    \prob{R_{i+1} | R_{i}}\cdot\prob{R_{i}}\ + \
    \prob{R_{i+1} | G_{i}}\cdot\prob{G_{i}}.
  \end{equation*}
  Independent of the preceding draws' outcomes, there will be $r+g-i$ marbles remaining after $i$ draws.
  In the case of $R_{i}$, there will be one less red marble, and one more green marble compared to $G_{i}$.
  Let $r_{i}$ ($g_{i}$) be the number of remaining red (green) marbles \emph{before} the preceding $i$-th marble was drawn.
  Depending on the color of the $i$-th marble, the number of red or green marbles \emph{after} the $i$-th draw differs by one and therefore,
  \begin{equation*}
    \prob{R_{i+1} | R_{i}} =
    \frac{r_{i}-1}{r_{i}+g_{i}-1}
    \quad \text{and} \quad
    \prob{R_{i+1} | G_{i}} =
    \frac{r_{i}}{r_{i}+g_{i}-1}.
  \end{equation*}
  Combining this, we get:
  \begin{equation*}
    \prob{R_{i+1}} =
    \frac{r_{i}-1}{r_{i}+g_{i}-1}\cdot\frac{r_{i}}{r_{i}+g_{i}} +
    \frac{r_{i}}{r_{i}+g_{i}-1}\cdot\frac{g_{i}}{r_{i}+g_{i}} =
    \frac{r_{i}(r_{i}-1+g_{i})}{(r_{i}+g_{i}-1)(r_{i}+g_{i})} =
    \frac{r_{i}}{r_{i} + g_{i}},
  \end{equation*}
  which is identical to $\prob{R_i}$ and thus by induction, $\prob{R_{i+1}}=\prob{R_{i}}= r / (r+g)$.
\end{proof}

A generalization of \cref{lem:polya-urn-reorder} to more colors is straightforward by fixing a color~$c$ of interest.
Then, color~$c$ takes the role of red, while all other colors are coalesced into green.

\begin{lemma}
  \label{lem:batch-sampling-correctness}
  \algSimple's and \algMulti's sampling algorithm for matrix~$D$ correctly simulates drawing and counting of $\ellHalf$ state pairs without replacement.
\end{lemma}

\begin{proof}
  When sampling from the matrix $D$ in \cref{sec:batch-processing}, we exploit this independence on previous draws when essentially reordering our sampling.
  We do so using two observations:
  \begin{itemize}
  \item Sampling first the initiators and then the responders is just reordering the draws.
  \item When sampling either of the two groups, we simply sample frequencies of the different colors.
    Due to \cref{lem:polya-urn-reorder} all permutations with the same frequency counts are equally likely to occur (\ie the probability of any sequence only depends upon the frequency counts of the colors sampled).
    We hence do not alter the distribution by sampling frequencies one color at a time. \qedhere
  \end{itemize}
\end{proof}
 
\clearpage

\fi

\inshort{\clearpage}
\section{Additional Benchmarks}
\label{apx:additional-benchmarks}

\begin{figure}[H]
    \includegraphics[width=\textwidth]{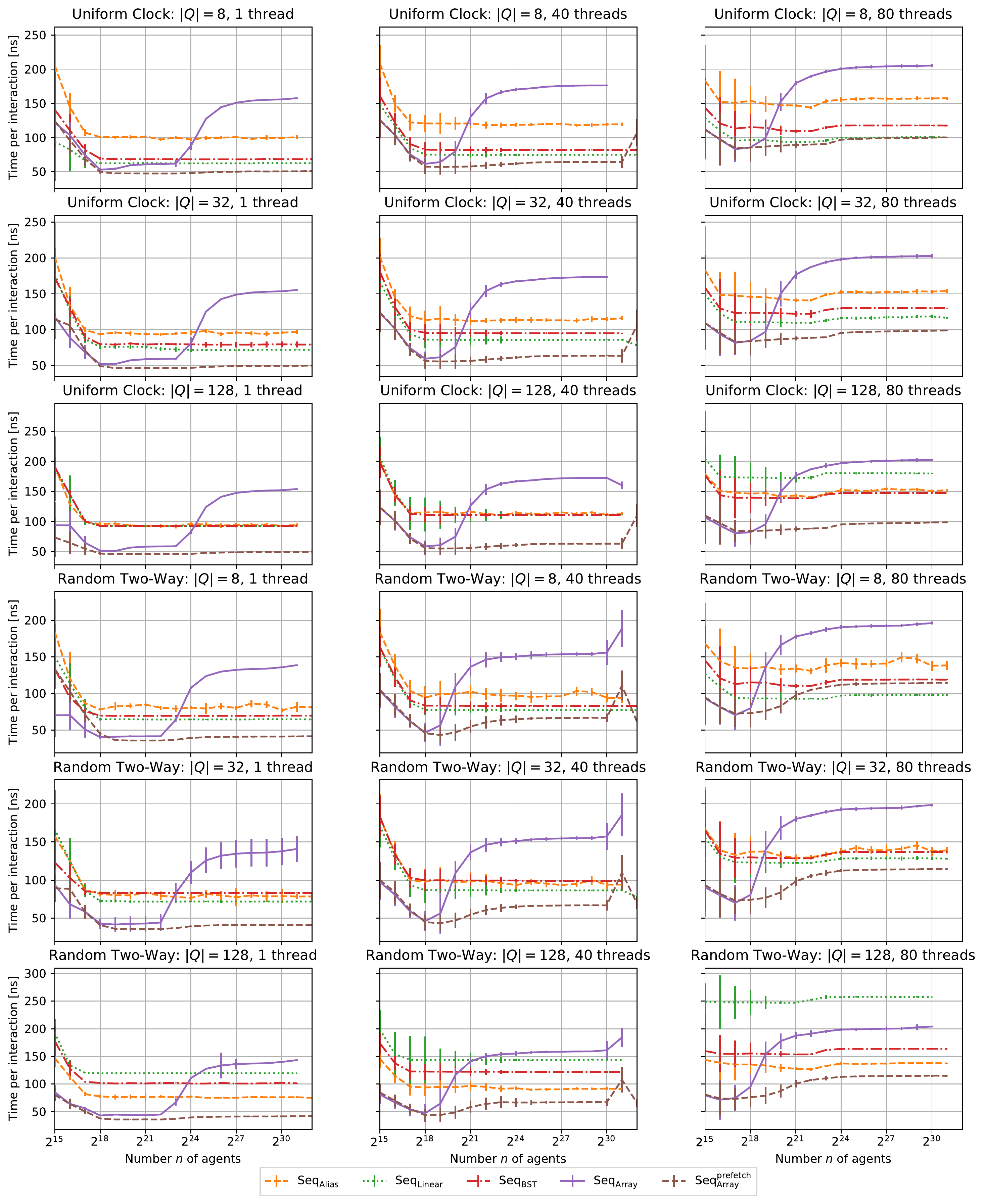}
    \caption{%
        Processing time per interaction as function of the number of agents $n$.
        Each series ends with the largest value of $n$ for which the median of the total processing time is below \SI{400}{\second}.
    }
    \label{fig:benchmark_grid_wo_batch}
\end{figure}

\inshort{\clearpage}
\begin{figure}[H]
    \includegraphics[width=\textwidth]{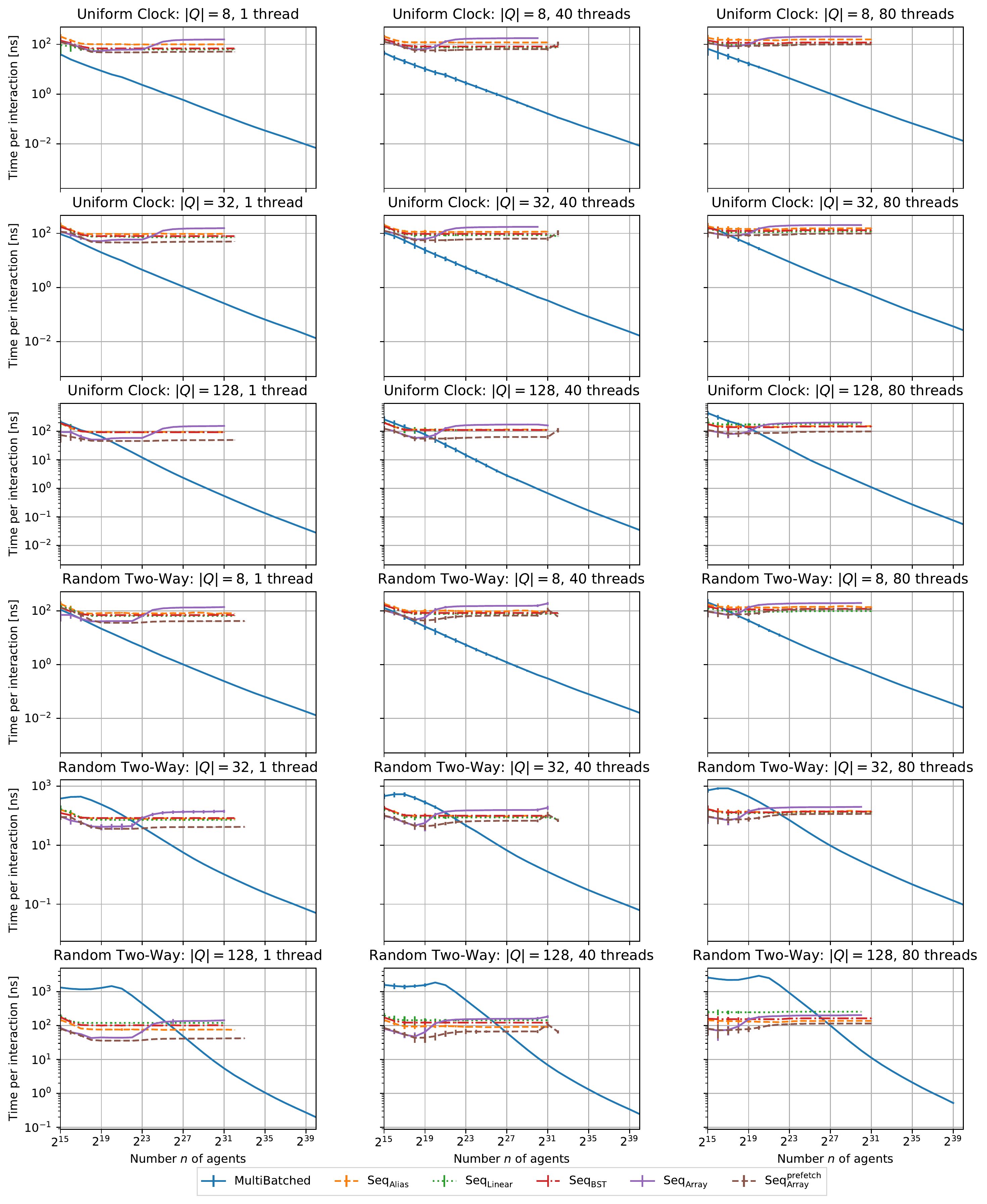}
    \caption{%
        Processing time per interaction as function of the number of agents $n$.
        Each series ends with the largest value of $n$ for which the median of the total processing time is below \SI{400}{\second}.
    }
    \label{fig:benchmark_grid_with_batch}
\end{figure}

\end{document}